\documentclass[10pt,a4paper]{article}

\usepackage{amsmath}

\usepackage[plain,noend]{algorithm2e}
\usepackage{amsmath,amssymb}
\usepackage{graphics,graphicx}
\usepackage{microtype}
\usepackage{mathtools}
\usepackage{bm}
\usepackage{color}

\usepackage{hyperref}

\usepackage{xcolor}
\hypersetup{
  colorlinks,
  linkcolor={red!50!black},
  citecolor={blue!50!black},
  urlcolor={blue!80!black}
}

\usepackage{kashdef}
\usepackage{kashThm}

\def\BIBDIR{./}
\def\PICDIR{./}

\usepackage{kashsty}

\usepackage[margin=1.25in]{geometry}
\usepackage[round,authoryear]{natbib}

\title{ 
  Non-asymptotic error controlled sparse high dimensional 
  precision matrix estimation
}
\author{Adam B Kashlak\\
Mathematical \& Statistical Sciences\\
University of Alberta\\
Edmonton, Canada,  T6G 2G1}

\begin{document}

\maketitle

\begin{abstract}
  Estimation of a high dimensional precision matrix 
  is a critical problem to many areas of statistics
  including Gaussian graphical models and inference on 
  high dimensional data.
  Working under the structural assumption of sparsity, we propose 
  a novel methodology for estimating such matrices 
  while controlling the false positive rate--percentage of 
  matrix entries incorrectly chosen to be non-zero.
  We specifically focus on false positive rates 
  tending towards zero with finite sample guarantees.
  This methodology is distribution free, but
  is particularly applicable to the problem of Gaussian network 
  recovery.  We also consider applications to 
  constructing gene networks in genomics data.
\end{abstract}

\section{Introduction}

Attempting to estimate the graphical structure of a high dimensional 
network is problemsome when edges are rare but the number of nodes
is large.  
In standard statistical classification problems, 
we fix a palatable false positive
rate and aim to recover as many true positives as possible.
Thus, we treat support recovery as a binary classification
problem.
For $p$ nodes, the number of potential edges to consider is on 
the order of $p^2$.  With a sample size $n<p$, there are too 
many parameters to accurately estimate.
However, when trying to classify edges as
significant or not under the assumption of sparsity, 
the assumption that most edges are not significant,  
if we have even a small false positive rate, then
this will result in many erroneous connections potentially obscuring 
follow up research.
For example, a genomics study considering the conditional correlation
structure of, say, 2000 genes will have to consider over two million 
potential edges.  A standard 1\% false positive rate will result in
tens of thousands of erroneous connections.

To address this issue, we consider the extreme estimation setting
of recovering the support of a precision matrix $\Omega$ in the high
dimensional setting, $p\gg n$, assuming $\Omega$ belongs to a 
class of sparse positive definite matrices, and for false positive
rates $\alpha\rightarrow0^+$.  Making use of debiased 
estimators \citep{BUHLMANN2013,JAVANMARD2014,VANDEGEER2014,JANKOVA2015},
the non-asymptotic results of \cite{KASHLAK_SPA2017}, 
and a clever subsampling methodology, we can achieve 
finite sample guarantees in this extreme setting.

There has been much past work on covariance and precision
matrix estimation much of which is summarized in the 
survey \cite{FAN2016}.  Most estimators for high dimensional
precision matrices are based on $\ell^1$ penalization 
including the graphical lasso \citep{GLASSO},
the CLIME and ACLIME estimators 
\citep{CAILIULUO2011,CAILIUZHOU2016}, and 
the debiased estimator of \cite{JANKOVA2015} used for 
constructing confidence sets and running hypothesis tests.
These articles all rely on high dimensional asymptotics, 
which is that estimation is successful in the limit as 
$p$ and $n$ grow to infinity together generally such that
$(\log p) /n$, or some variant thereof, is $o(1)$.
Our work differs as it considers guaranteed results for
controlling the support recovery of $\Omega$ as the 
false positive rate $\alpha\rightarrow0$ for
fixed finite $p\gg n$, which is asymptotic in $\alpha$, 
a controllable tuning parameter, instead of in $n$ and $p$,
which are generally fixed in the real world experimental
setting.

Our method can be deconstructed into two steps.  The first is to 
find a bias-corrected initial estimate $\hat{\Omega}_0$ for the 
precision matrix $\Omega$.  The second is to construct 
a ball in the operator norm topology corresponding to the target 
false positive rate and then
search this ball for a sparse estimator.
In actuality, given a false positive rate of $\alpha$,
we step down in a binary fashion constructing successive
false positive rates of $\gamma^{-1},\,\gamma^{-2},\ldots$ 
and a sequence of estimators tending towards one with the
desired $\alpha$.
Depending on the sample size, we can 
randomly partition our data set into smaller sets,
apply this two-step procedure to each subsample in parallel,
and combine the results for improved performance.
In step two, we use a binary search procedure which 
converges rapidly minimizing the computational burden. 

Our search methodology is effectively a method for controlled
thesholding of the initial debiased estimator.  The literature
on thresholding for precision matrices is quite light when 
compared to thresholding covariance matrices 
\citep{BICKELLEVINA2008,BICKELLEVINA2008A,ROTHMAN2009,
CAILUI2011,KASHLAK_SPA2017}.  The challenge is that unlike
the unbiased empirical covariance estimator, there is no 
unbiased empirical precision estimator.  Debiasing such 
precision matrix estimators allows for thresholding as 
mentioned in \cite{JANKOVA2015}.  Furthermore, 
\cite{JANKOVA2015} proposed an entrywise thresholding method
given some false positive rate $\alpha$, which while 
working well in some high dimensional asymptotic sense,
proved to not achieve the target false positive rate
on simulated data.

\section{Methodology}

\subsection{Notation}
\label{sec:notation}

In this article, we primarily make use of the family of 
Schatten norms.  For $\Omega\in\real^{p\times p}$ with 
$i,j$th entry $\Omega_{i,j}$ and eigenvalues $\lmb_1,\ldots,\lmb_p$, 
we define the trace norm 
$\norm{\Omega}_1 = \sum_{i=1}^p \abs{\lmb_i}$, 
the Hilbert-Schmidt norm
$ 
  \norm{\Omega}_2 = (\sum_{i=1}^p \abs{\lmb_i}^2)^{1/2} = 
  (\sum_{i,j=1}^p \Omega_{i,j}^2)^{1/2},
$ 
and the operator norm 
$\norm{\Omega}_\infty = \max_{i=1,\ldots,p}\abs{\lmb_i}$.
We will also use the standard $\ell^p$ norms applied to vectors
in $\real^p$
and entrywise to matrices in $\real^{p\times p}$.

The algorithm in Section~\ref{sec:algorithm} makes use of 
hard thresholding.  For $\Omega\in\real^{p\times p}$, we
define the hard thresholding operator 
$\varphi:\real^{p\times p}\times \real^+ \rightarrow \real^{p\times p}$,
which returns a matrix with $i,j$th entry
$$
  \varphi(\Omega; t)_{i,j} = 
  \left\{
  \begin{array}{ll}
    0 & \text{if }\abs{\Omega_{i,j}}< t\\
    \Omega_{i,j} & \text{if } \abs{\Omega_{i,j}}\ge t
  \end{array}
  \right.
$$
which simply removes entries from $\Omega$ with magnitude less 
than $t$.

\subsection{Initial Estimator}
\label{sec:initEst}

Let $X_1,\ldots,X_n\in\real^p$ be an \iid collection of mean zero
random vectors with common unknown positive definite 
covariance matrix $\Sigma$ and corresponding precision matrix 
$\Omega = \Sigma^{-1}$.  Of course, we require the covariance
and its inverse 
to exist, but put no further assumptions on the distribution 
of the $X_i$ at the moment except to assume that $\Omega$ is sparse.  
Specifically, 
\begin{multline*}
  \Omega \in \mathcal{U}(\kappa,\delta) = \left\{ 
    A \in \real^{p\times p} \,:\,
    \right.\\\left. 
    A=\TT{A},~
    \sum_{i=1}^p {\bf1}(A_{i,j}\ne0)\le \kappa ~\text{for }j=1,\ldots,p,
    ~\text{if }A_{i,j}\ne0\text{ then }\abs{A_{i,j}}\ge\delta 
  \right\}
\end{multline*}
where $\mathcal{U}(\kappa,\delta)$ is the class of sparse 
$\real^{p\times p}$
matrices with no more than $\kappa$ entries in each row or column
non-zero and with non-zero entries bounded away from zero
allowing them to be detectable.  
In practice, we attempt to recover the support of the normalized 
matrix $\Omega$ with diagonal entries of 1 to avoid scale issues.

Our approach is similar to \cite{KASHLAK_SPA2017} who attempt to 
recover the support of a covariance matrix $\Sigma$ by starting 
with an initial estimator, constructing a confidence ball, and
searching this ball for a sparser estimator.  However,
whereas this preceding work can use the empirical estimator,
$\hat{\Sigma}=n^{-1}\sum_{i=1}^nX_i\TT{X_i}$, as an unbiased
initial estimator for controlled thresholding, 
we cannot construct an unbiased estimator
for the precision matrix $\Omega$ in the $p>n$ setting. 

A standard estimator for sparse precision matrices is the
graphical lasso \citep{GLASSO}, which is based on an $\ell^1$
penalized maximum likelihood under the Gaussian distribution:
$$
  \hat{\Omega}^\text{GL} = 
  \argmin{ \Theta\in\real^{p\times p} }\left\{
    \tr{\hat{\Sigma}\Theta} - \log\det(\Theta)
    + \lmb\sum_{i=1}^p\sum_{j\ne i}\abs{\Theta_{i,j}}
  \right\}
$$
for some tuning parameter $\lmb>0$.  The graphical lasso is 
debiased in \cite{JANKOVA2015}, which uses a correction
factor based on the subgradient of the above optimization
extending the work of \cite{VANDEGEER2014} that debiases 
the classic lasso estimator for linear regression.  
The resulting debiased estimator is
$$
  \hat{\Omega}^\text{deGL} = 
  2\hat{\Omega}^\text{GL} - 
  \hat{\Omega}^\text{GL}\hat{\Sigma}\hat{\Omega}^\text{GL}.
$$

Alternatively, the CLIME method \citep{CAILIULUO2011} solves
a constrained $\ell^1$ optimization problem to find $\hat{\Omega}^{CL}$.
$$
  \min_{\Theta\in\real^{p\times p}} 
  \norm{\Theta}_{\ell^1}
  \text{ such that }
  \max_{i,j}\abs{ \hat{\Sigma}\Theta - I_p }_{i,j} \le \lmb
$$
with a more sophisticated version ACLIME \citep{CAILIUZHOU2016} adapting
to individual entries.

A different regularized estimator is the ridge estimator, 
$$
  \hat{\Omega}^\text{RD} =
  \argmin{\Theta\in\real^{p\times p}}\left(
    \norm{ \hat{\Sigma}\Theta - I }_2^2 + \lmb\norm{\Theta}_2^2
  \right),
$$
whose closed form solution is 
$\hat{\Omega}^{RD}=(\Sigma+\lmb I_p)^{-1}$.
Consider the singular value decomposition $X = UD\TT{V}$ 
where $U\in\real^{n\times r}$, $D\in\real^{r\times r}$, and
$V\in\real^{p\times r}$ for $r=\min( n,p )$.
In \cite{BUHLMANN2013}, a bias corrected estimator for 
ridge regression is proposed.  This is achieved by finding
some other estimator, which is used to correct for the
projection bias.  We can apply the same method to the ridge 
precision matrix estimator to get
$$
  \hat{\Omega}^\text{deRD} = 
  \hat{\Omega}^\text{RD} + P_0\tilde{\Omega}
$$
where $\tilde{\Omega}$ is some other estimate for $\Omega$
and $P_0 = V\TT{V} - \text{diag}(V\TT{V})$ where
$V\TT{V}$ is
the projection in $\real^p$ onto an $n$ dimensional subspace
spanned by the columns of $X$, and $P_0$ is that projection 
with the diagonal entries set to zero.  
\cite{BUHLMANN2013} uses the graphical lasso 
estimator for
$\tilde{\Omega}$.

\subsection{Binary Search}
\label{sec:algorithm}

Given an initial estimator from the previous setting
henceforth denoted $\hat{\Omega}_0$
for simplicity of notation and some $\gamma>1$ such that
$\alpha = \gamma^{-m}$ for some positive integer $m$ corresponding
to the number of iterates of the below algorithm, 
we aim to construct 
an estimator
$\hat{\Omega}_m$ with false positive rate $\alpha=\gamma^{-m}$ 
by carefully thresholding the 
entries in $\hat{\Omega}_0$ so that 
$$
  \frac{
    \abs*{\left\{
      (i,j) \,:\, 
      (\hat{\Omega}_m)_{i,j}\ne0\text{ and }
      (\Omega)_{i,j}=0,~i\ne j 
    \right\}}
  }{
    \abs*{\left\{
      (i,j) \,:\, (\Omega)_{i,j}=0,~i\ne j
    \right\}}
  } \approx \alpha,
$$
which is that the desired false positive rate is achieved.
The two extreme estimators are the initial estimator 
$\hat{\Omega}_0$ and the diagonal matrix 
$\hat{\Omega}_\infty$ with diagonal entries coinciding with 
$\hat{\Omega}_0$ and off-diagonal entries set to zero.
These correspond to the 100\% and 0\% false positive 
cases, respectively.  For the remainder, we normalize
$\hat{\Omega}_0$ to have unit diagonal thus making 
$\hat{\Omega}_\infty=I_p$.

For values of $\alpha$ tending towards zero, we consider the 
operator norm balls centred at $I_p$ being
$
  B_\alpha = \{
    \Pi\in\real^{p\times p}\,:\,
    \norm{ \Pi - I_p }_\infty \le
    \norm{ \hat{\Omega}_m - I_p }_\infty 
  \}
$ 
where $\norm{\cdot}_\infty$ refers to the operator norm 
for $\ell^2(\real^p)\rightarrow\ell^2(\real^p)$, which is
the principal eigenvalue.  This motivates the following 
algorithm:

Constructing an error controlled estimator 
\begin{tabbing}
  \qquad \enspace Begin with an initial estimator from Section~\ref{sec:initEst}
         denoted $\hat{\Omega}_0$.\\ 
  \qquad \enspace Given the $s$th iterated estimator $\hat{\Omega}_{s}$,
    we construct the $(s+1)$th estimator\\
    \qquad\qquad Compute $r_s = \norm{\hat{\Omega}_{s}-I_p}_\infty$.\\
    \qquad\qquad Compute $r_s'= r_s \gamma^{-1/2}$.\\
    \qquad\qquad 
          Find $t_{s+1}=\min(t)$ such that 
          $\norm{\varphi(\hat{\Omega}_{s};t)-I_p}_\infty\le r_s'$\\
    \qquad\qquad ~
          where $\varphi$ is the hard threshold operator.\\
    \qquad\qquad Set 
      $\hat{\Omega}_{s+1} = \varphi(\hat{\Omega}_{s};t_{s+1})$\\
  \qquad \enspace Repeat step 2 with $\hat{\Omega}_{s+1}$ until
                  $s=m$\\
\end{tabbing}

In this algorithm, we quickly locate the densest estimator close 
to $I_p$, alternatively being the sparsest estimator close to 
$\hat{\Omega}_{0}$, by using a binary search procedure.
Given the $s$th iterated matrix as a starting point,
we set the smallest half of the non-zero entries in magnitude in 
$\hat{\Omega}_{s}$ to zero, compute the distance to $I_p$
and then if the distance is greater than $r_s'$, we remove half 
of the remaining entries whereas if the distance is less than $r_s'$,
we reintroduce half of the removed entries.

\begin{remark}
  \label{rk:initChoice}
  For choice of $\gamma>1$ and number of steps $m$, values of 
  $\gamma$ closer to 1 result in smaller steps adding stability
  but requiring a larger $m$ to achieve a small false positive rate.
  In the simulations of Section~\ref{sec:simData}, we use $\gamma=2$. 
  
  For the choice of initial estimator from Section~\ref{sec:initEst}, 
  the best performance in
  simulated data experiments on multivariate normal data
  was achieved by
  beginning the above procedure with the debiased
  graphical lasso estimator of \cite{JANKOVA2015}. 
  This is mainly because their estimator follows the requirements
  of the theorems in the subsequent section as long as the 
  data under analysis is sub-Gaussian.
  However, good performance is still observed when the 
  data is sub-exponential as can be seen in the supplementary
  material.
\end{remark}

\begin{remark}
  A similar search algorithm is proposed in \cite{KASHLAK_SPA2017}.
  The main differences are that in the cited work, the operator
  norm balls are confidence sets centred about the empirical
  covariance estimator rather than centred about the identity
  matrix $I_p$.  Furthermore, the radii are reduced by a factor
  of $\gamma^{-1}$ in that previous work due to the low rank structure 
  of the empirical covariance estimator.  Specifically, 
  $\xv\norm{\hat{\Sigma}} = O(p)$ rather than $O(p^{1/2})$.
\end{remark}

\subsection{Theoretical Guarantees}

The reason for shrinking the radius by $\gamma^{-1/2}$   
in the above algorithm comes from applying tools from random matrix
theory \citep{TAO2012} to sparse matrices---see the supplementary
material for proofs.
The following result states that 
when $\Omega$ is sufficiently sparse and $p\gg n$, we can reduce 
the false positives by $\gamma^{-1}$ by shrinking the radius of an operator
norm ball centred around the $p$-dimensional identity matrix $I_p$
by $\gamma^{-1/2}$.

\begin{theorem}[Controlled False Positives]
  \label{thm:falsepos}
  Let $\Omega\in\mathcal{U}(\kappa,\delta)$  
  with $\kappa=O(p^{\nu})$ for $\nu<1/2$, and
  $\norm{\Omega}_\infty=o(p^{1/2})$. 
  For some false positive rate $\alpha=\gamma^{-s}$ with $s\in\integer^+$, 
  let the bias of the initial estimator be
  $
    \norm{\text{bias}(\hat{\Omega})}_\infty =
    \norm{\xv\hat{\Omega}-\Omega}_\infty = o(\gamma^{s}p^{1/2}). 
  $
  Then, 
  \begin{align*}
    \mathrm{(a)}~~&
    { \frac{ 
      \xv\norm{\hat{\Omega}_{s+1}-\Omega}_\infty 
    }{
      \xv\norm{\hat{\Omega}_{s}-\Omega}_\infty
    }} = \frac{1 + o(1)}{\gamma^{-1/2} + o(1)}
    &\mathrm{(b)}~~&
    {\frac{ 
	  \xv\norm{\hat{\Omega}_{{s+1}}-I_p}_\infty 
  	}{
 	  \xv\norm{\hat{\Omega}_{{s}}-I_p}_\infty
    }} = \frac{1 + o(1)}{\gamma^{-1/2} + o(1)}
  \end{align*}
\end{theorem} 

\begin{remark}
  To control the false positive rate, we require a few assumptions
  in the theorem statement.  Namely, the number of non-zero entries
  per row $\kappa$ and the operator norm of $\Omega$ cannot grow
  at a rate faster than $p^{1/2}$ as $p$ increases.  In the 
  simulations of Section~\ref{sec:simData}, we have the much 
  nicer setting where these quantities remain bounded as $p$ increases.
\end{remark}
\begin{remark}
  The accuracy of the conclusion of Theorem~\ref{thm:falsepos}
  is improved as $p$ increases for a fixed sample size $n$.
  This fact motivates the subsampling methodology in 
  Section~\ref{sec:subsamp}.  
  Secondly, as the false positive rate decreases, the
  bias is allowed to be larger.  Hence, we can choose $\gamma^{-s}$
  based on $p$.  This is because as $s$ increases, we threshold
  more aggressively effectively fighting against increases in the
  bias.
\end{remark}

As noted in Remark~\ref{rk:initChoice}, we chose the method of 
\cite{JANKOVA2015} for our simulations in Section~\ref{sec:simData}.
It is shown in their work that this debiased estimator has the following 
maximal entrywise bias,
$b_\text{max} = 
\max_{i,j}\{\text{bias}(\hat{\Omega})\} = 
O(\kappa\log(p)n^{-1/2})$
or $=O(\kappa^{3/2}\log(p)n^{-1/2})$ depending on specific assumptions
on the sub-Gaussian nature of the data.
Hence, 
$\norm{\text{bias}(\hat{\Omega})}_\infty=o(b_\text{max} p^{1/2}\gamma^{-s})$.
As we have control over the false positive rate $\gamma^{-s}$, we can 
choose this to satisfy the conditions of the above theorem.
Namely, considering false positive rates less than $p^{-1/2}$
or $p^{-3/4}$, which are generally of more interest than large
false positive rates.

Further considering the debiased estimator described in 
\cite{JANKOVA2015} and given the assumptions made in that
article, we can recover the support of the matrix asymptotically
as $n,p$ increase and $\alpha\rightarrow0$.
Indeed, Lemma~9 from \cite{JANKOVA2015} and equivalently
Theorem~1 from \cite{RAVIKUMAR2011} assume an 
{irrepresentability condition} common in the lasso literature
and get convergence rates of the graphical lasso estimator 
depending on the tail behaviour of the random vectors $X_1,\ldots,X_n$.
Thus, the graphical lasso estimator asymptotically recovers the support.
Further assuming sub-Gaussian tails for the $X_i$, 
\cite{JANKOVA2015} shows that the remainder term in the
debiased graphical lasso estimator is asymptotically negligible.
Thus, thresholding the debiased estimator will in turn re-recover
the support.  
Similarly, the literature on sparse covariance matrix estimation 
generally requires sub-Gaussian tails for asymptotic support 
recovery 
\citep{BICKELLEVINA2008,BICKELLEVINA2008A,ROTHMAN2009,
CAILUI2011,KASHLAK_SPA2017}.

Beyond distributional assumptions, a quick calculation can
demonstrate that the true positive probability is 
necessarily greater than the false positive probability.
For some threshold $t_s\in[0,1]$ corresponding to a false
positive rate of $\alpha = \gamma^{-s}$, we have assuming
symmetry of the distribution of $\hat{\Omega}_{i,j}$ that
\begin{align*}
  \mathrm{P}(\text{true positive}) 
  &= 
  \mathrm{P}\left(
    \abs{\hat{\Omega}_{i,j}} > t_s\,\mid\,
    \Omega_{i,j}\ne0
  \right) 
  \\&=
  2\mathrm{P}\left(
    {\hat{\Omega}_{i,j}}-\Omega_{i,j} > t_s-\Omega_{i,j}\,\mid\,
    \Omega_{i,j}\ne0
  \right) 
  \\&=
  \mathrm{P}\left(
    \abs{\hat{\Pi}_{i,j}} > t_s'
  \right)
  = \mathrm{P}(\text{false positive for }t_s')
\end{align*}
which is that the probability of a true positive 
at threshold $t_s$
corresponds to the probability of a false positive
at threshold $t_s'>t_s$.  Thus, this method is 
guaranteed to return at least as many true positives
proportionally as false positives and generally,
as will be seen in Section~\ref{sec:simData},
performs much better.

\subsection{Subsampling}
\label{sec:subsamp}

As Theorem~\ref{thm:falsepos}, becomes more accurate for 
large $p$, we can run the above methodology
in parallel by randomly partitioning the sample of size $n$
into subsamples of size $n/k$.
As a result, we run our method in parallel $k$ times returning
estimators $\hat{\Omega}^{(1)},\ldots,\hat{\Omega}^{(k)}$.
For a false positive rate of $0<\alpha\ll1$ and the assumption 
from Theorem~\ref{thm:falsepos}
that $\Omega\in\mathcal{U}(\kappa,\delta)$ with 
$\kappa=O(p^\nu)$ and $\nu<1/2$, then the expected number of 
false positive recoveries is 
$\alpha[p(p+1)/2-Cp^{1+\nu}] = O(\alpha p^2)$.  
After splitting the sample into $k$ disjoint pieces, 
we can run the above algorithm in parallel to construct
$k$ independent estimators $\hat{\Omega}^{(1)},\ldots,\hat{\Omega}^{(k)}$.
The probability
of recovering the same false positive entry in $d$ or more of the 
$\hat{\Omega}^{(i)}$ is a binomial tail sum: 
$\sum_{i=d}^k {k\choose i}\alpha^i(1-\alpha)^{n-i}$.
When $\alpha < 0.001$, then choosing $d=2$ is 
generally sufficient.
Indeed, we can bound this binomial tail sum by 
the Kullback-Leibler divergence. 
For $B\dist\distBinom{k}{\alpha}$, 
\begin{multline*}
  \mathrm{P}\left({B\ge d}\right) 
  \le \exp\left\{
      -{d} \log\left( \frac{d/k}{\alpha} \right)  
      -({k-d})\log\left(\frac{1-d/k}{1-\alpha}\right) 
  \right\}\\
  \le k^k\left(\frac{\alpha}{d}\right)^d
       \left(\frac{1-\alpha}{k-d}\right)^{k-d}
  \le \frac{k^k}{d^d(k-d)^{k-d}}\alpha^d
  \le 2^k\alpha^d.
\end{multline*}
If we want a target false positive rate of $2^{-10}$, 
we can choose $\alpha$ such that 
$2^k\alpha^d = 2^{-10}$ or $\alpha = 2^{-(10+k)/d}$
to achieve the same false positive rate as when not
subsampling.  For example, if we take $d=2$ and keep $k<10$, then we 
can relax the false positive rate for each individual 
subsampled estimator.

This addition of subsampling to the methodology allows for 
faster runtimes, as the $k$ estimators can be computed in 
parallel with fewer iterations, and also increases accuracy
as the target false positive rate decreases below $p^{-1}$
as will be seen in Section~\ref{sec:simData}.

\section{Applications}

\subsection{Simulated Data}
\label{sec:simData}

We test our methodology for three different graphical structures.
In the first case,
$\Omega$ tridiagonal--i.e. $\Omega_{i,i}=1$, $\Omega_{i,j}=1/3$
if $\abs{i-j}=1$, $\Omega_{i,j}=0$ if $\abs{i-j}>1$.  This is 
the sparsest connected graph structure with $p-1$ edges on $p$
nodes.
In the second case,
$\Omega$ represents a binary tree graph with depth $d$
resulting in $d(d+1)/2$ nodes and $d(d-1)$ edges, which is 
roughly $2p$ edges for $p$ nodes.  Similar to case one,
diagonal entries are set to $1$ and non-zero off-diagonal entries
are set to $1/3$. 
In the third case,
$\Omega$ represents a disjoint collection of $k$ complete graphs
on $d$ nodes.  Thus, there are $kd$ nodes and $kd(d+1)/2$ edges,
and $\Omega$ is a block diagonal matrix with blocks 
$\frac{1}{3}{\bf1}_{d} + \frac{2}{3}I_{d}$ with ${\bf1}_d$ being
the $d\times d$ matrix of all $1$'s and $I_d$, the ${d\times d}$
identity matrix.

For simulations, the sample size was fixed at $n=50$.
The dimensions considered are $p=496,1128,2080$ as they
correspond to binary trees of depth $31,47,64$, respectively,
and correspond to block diagonal matrices with 
$62,141,260$ $(8\times8)$-blocks, respectively. In all nine 
cases considered, the non-zero off-diagonal entries are all 
set to $1/3$ with the diagonal entries set to $1$.

In Figure~\ref{fig:fptpPlot}, we plot the empirical 
true and false positive rates on the vertical axis against
the target false positive rate on the horizontal axis.  The
target false positive rate acts as a tunable parameter.
Three types of matrices, three dimensions, and three values
of the graphical lasso regularization parameter $\lmb = 0.5,1,2$ are 
considered.
Ideally, the observed false positive rates will be close
to the solid black line, which is where the observed and 
target rates coincide.  This occurs as the dimension
increases in all three cases.  For the sparser models--tridiagonal
and binary tree structures--we are able to maintain a true
positive rate above 20\% as the false positive rate is taken to zero.
Performance is much worse when trying to recover the 
block diagonal matrix.  For the lowest dimension, $p=496$,
setting the penalization parameter $\lmb=0.5$ gives the best
results whereas in the higher dimensions $\lmb=1$ gives the 
best performance albeit only slightly.  Hence, it appears that
for choosing an initial estimator, the graphical lasso penalization parameter 
can be set to $1$ regardless of dimension and true $\Omega$ 
for this methodology.
On this note, we also considered $\lmb=0.25$
but the performance was terrible and thus it is not included
in the figures.  For an alternative look at the simulations,
receiver operating characteristic curves are included in the 
supplementary material.
The same simulations were also run on multivariate Laplace data
with details in the supplementary material.  In short, 
choosing $\lmb=1$ allowed for the empirical false positive 
rate to stay close to the target rate whereas 
$\lmb=2$ did not perform as well.

In Figure~\ref{fig:subs}, we consider the tridiagonal and 
binary tree models estimated by the standard approach as before
and by subsampling as in Section~\ref{sec:subsamp} with $k=2$.
By subsampling, we can extend the false positive rate from 
around $\alpha=p^{-1}$, where the standard method begins to 
breakdown, all the way to around $\alpha=p^{-2}$.  Hence, 
we can have a controlled reduction in the empirical false positive 
rate to effectively zero, as the precision matrix only has $p^2$ 
entries, using the subsampling method.

In Table~\ref{tab:glsclm}, we tabulate empirically achieved true
and false positive rates from the cross validated graphical lasso 
\citep{GLASSOPACK} and 
ACLIME estimators \citep{FASTCLIME}. These estimators were 
computed by splitting the $n=50$ sample in half and optimizing
the tuning parameters with respect to the operator norm distance.
The graphical lasso estimator's achieved false positive was roughly 
halved as the dimension doubled.
The ACLIME estimator much more aggressively penalized the 
matrix entries.  
The threshold method of \cite{JANKOVA2015} was also tested, which
for a false positive rate $\alpha$, sets the individual entries
in the debiased estimator $\hat{\Omega}$ to zero if
$$
  \hat{\Omega}_{i,j} < 
  \Phi^{-1}\left(1-\frac{\alpha}{p(p-1)}\right)
  \frac{\hat{\Sigma}_{i,j}}{\sqrt{n}}
$$
where $\hat{\Sigma}$ is the empirical covariance matrix and $\Phi(\cdot)$ is
the cumulative distribution function for the standard normal distribution.  
A figure in line with Figure~\ref{fig:fptpPlot} for this method 
is included in the supplementary 
material.  In short, this approach should work asymptotically as 
$n,p\rightarrow\infty$.  In
our simulations, the empirically achieved false positive rate was not 
close to the target false positive rate.

\begin{figure}
  \begin{center}
  \includegraphics[width=0.32\textwidth]{\PICDIR/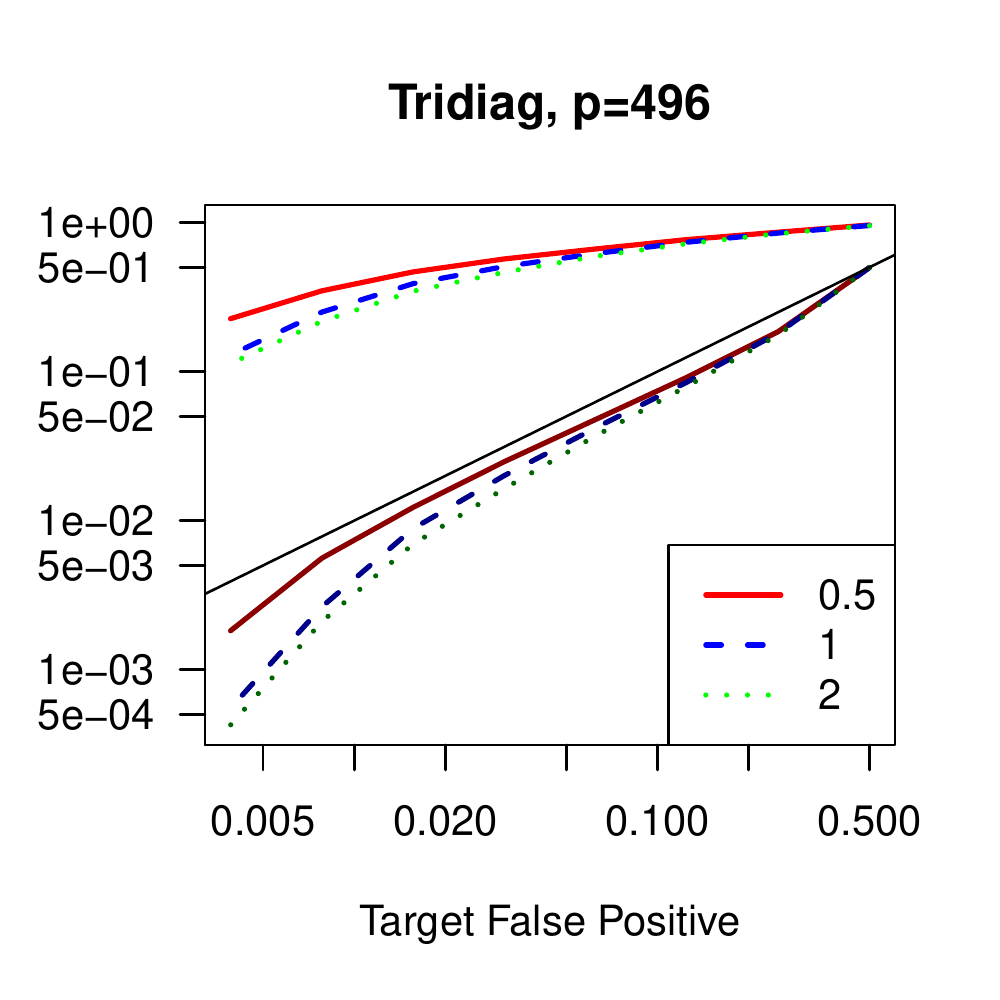}
  \includegraphics[width=0.32\textwidth]{\PICDIR/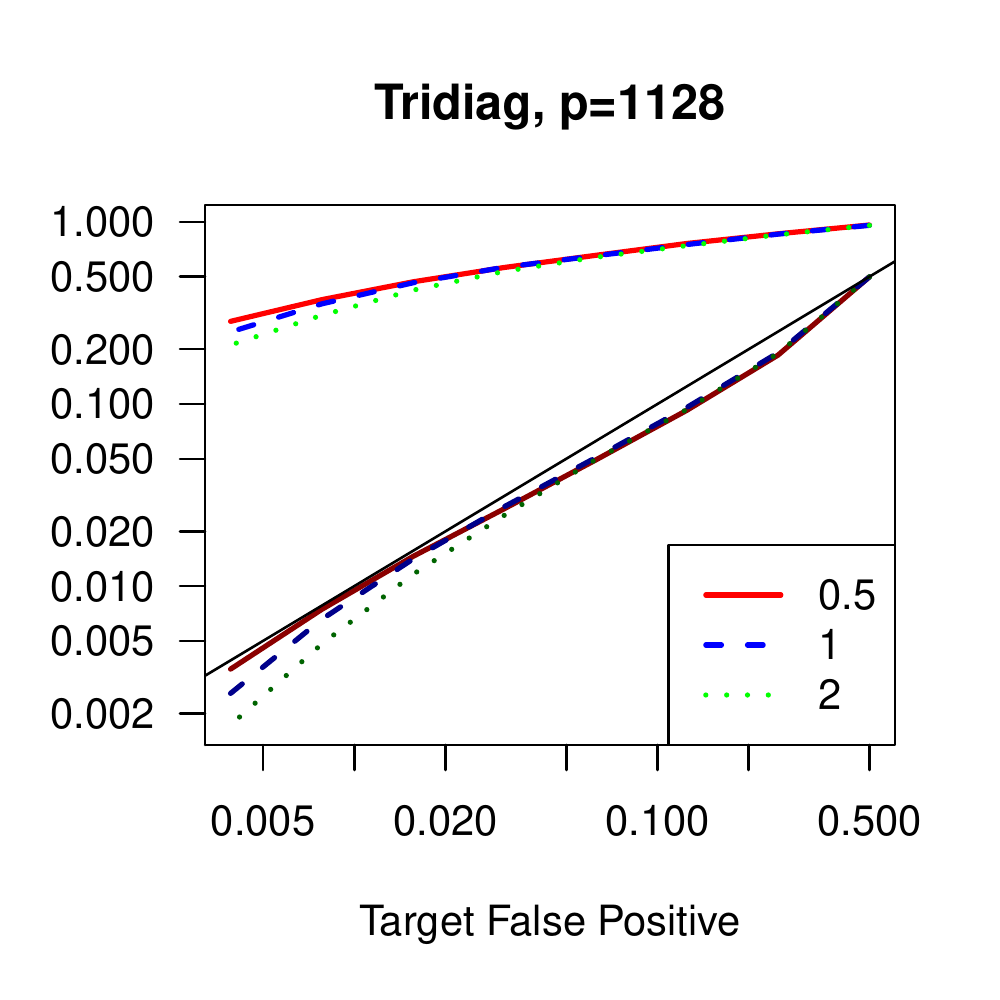}
  \includegraphics[width=0.32\textwidth]{\PICDIR/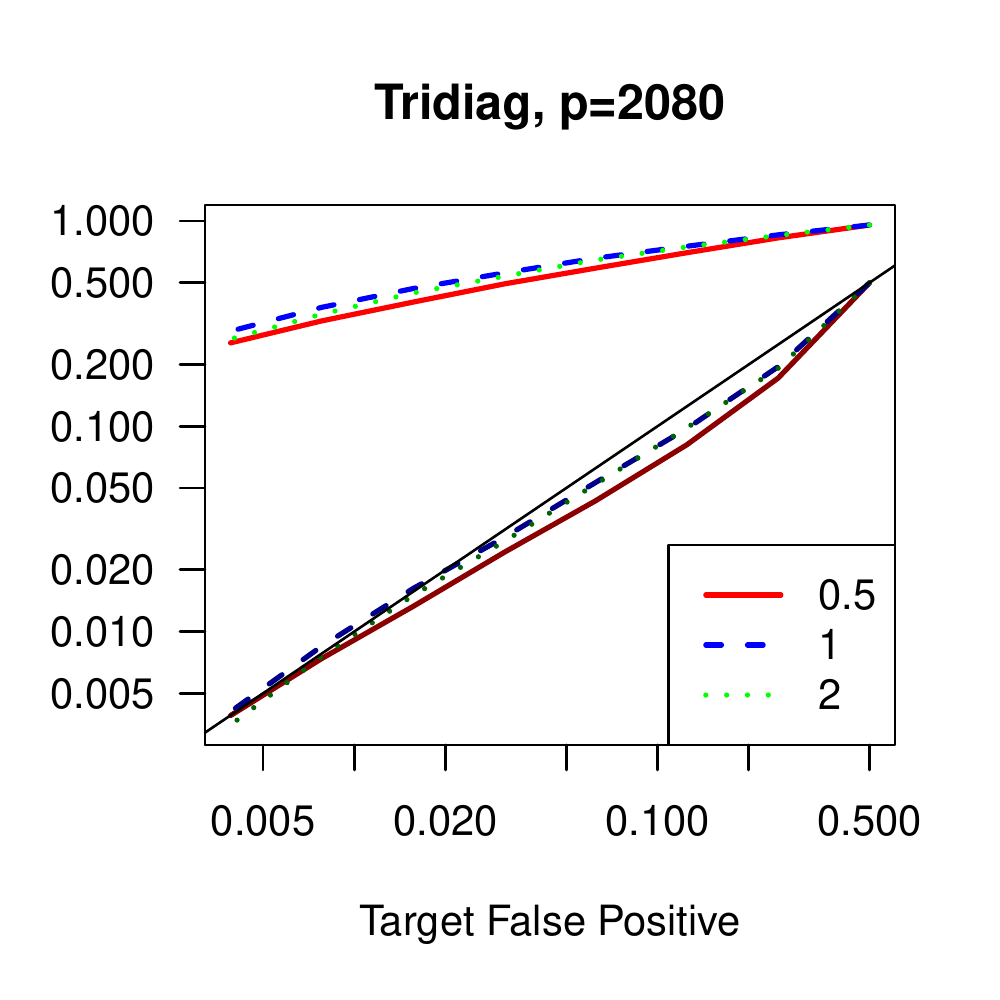}
  \includegraphics[width=0.32\textwidth]{\PICDIR/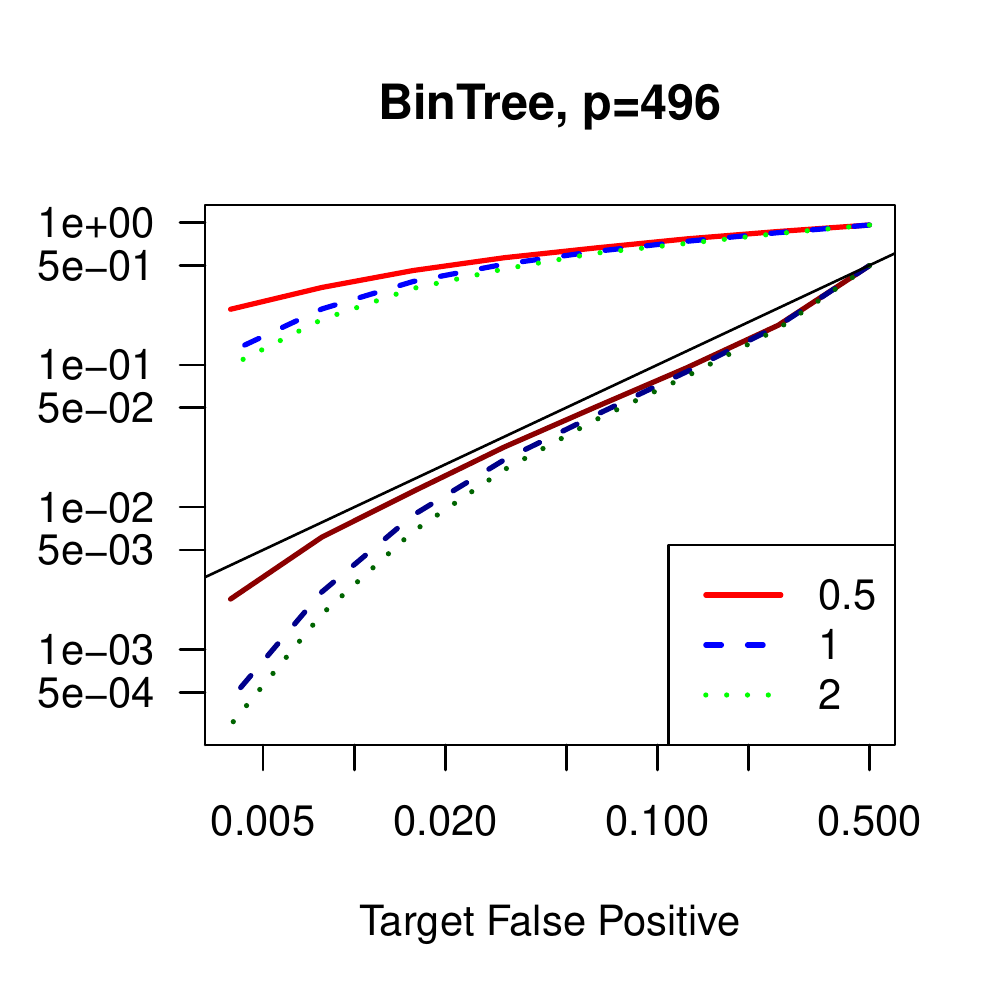}
  \includegraphics[width=0.32\textwidth]{\PICDIR/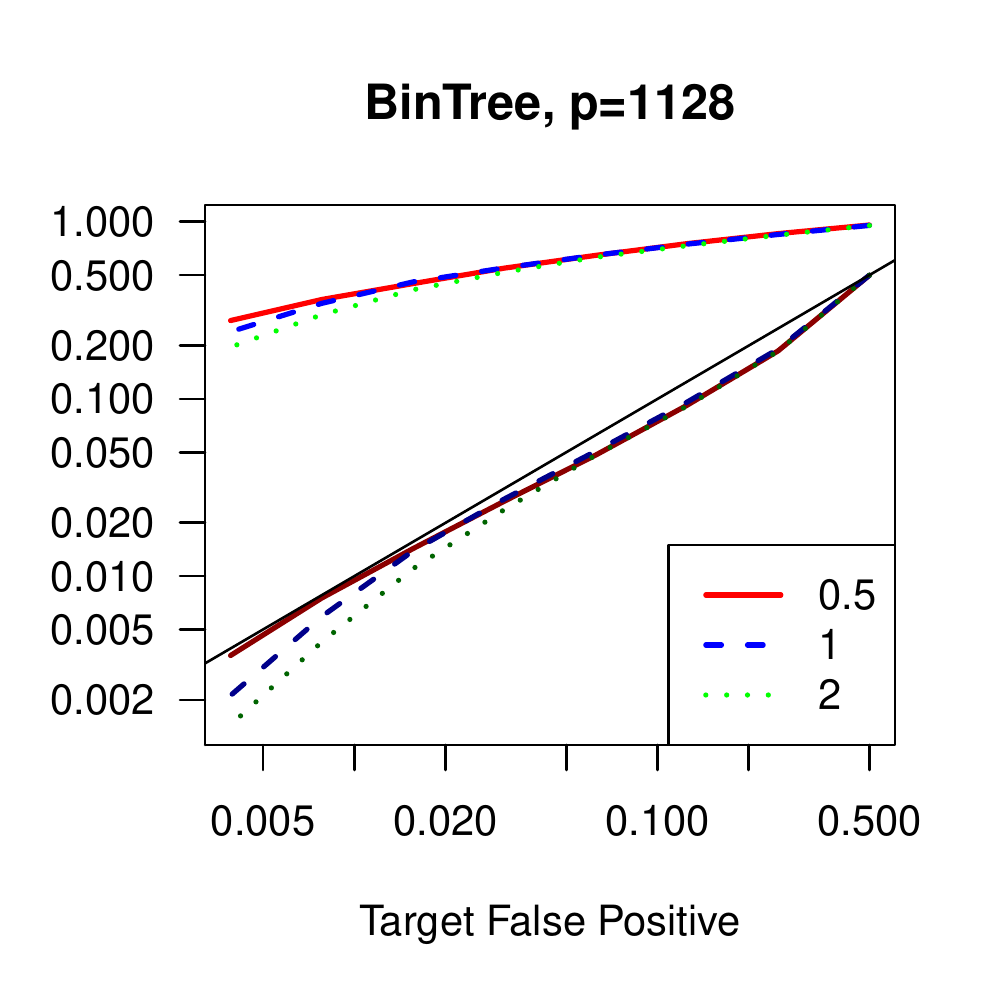}
  \includegraphics[width=0.32\textwidth]{\PICDIR/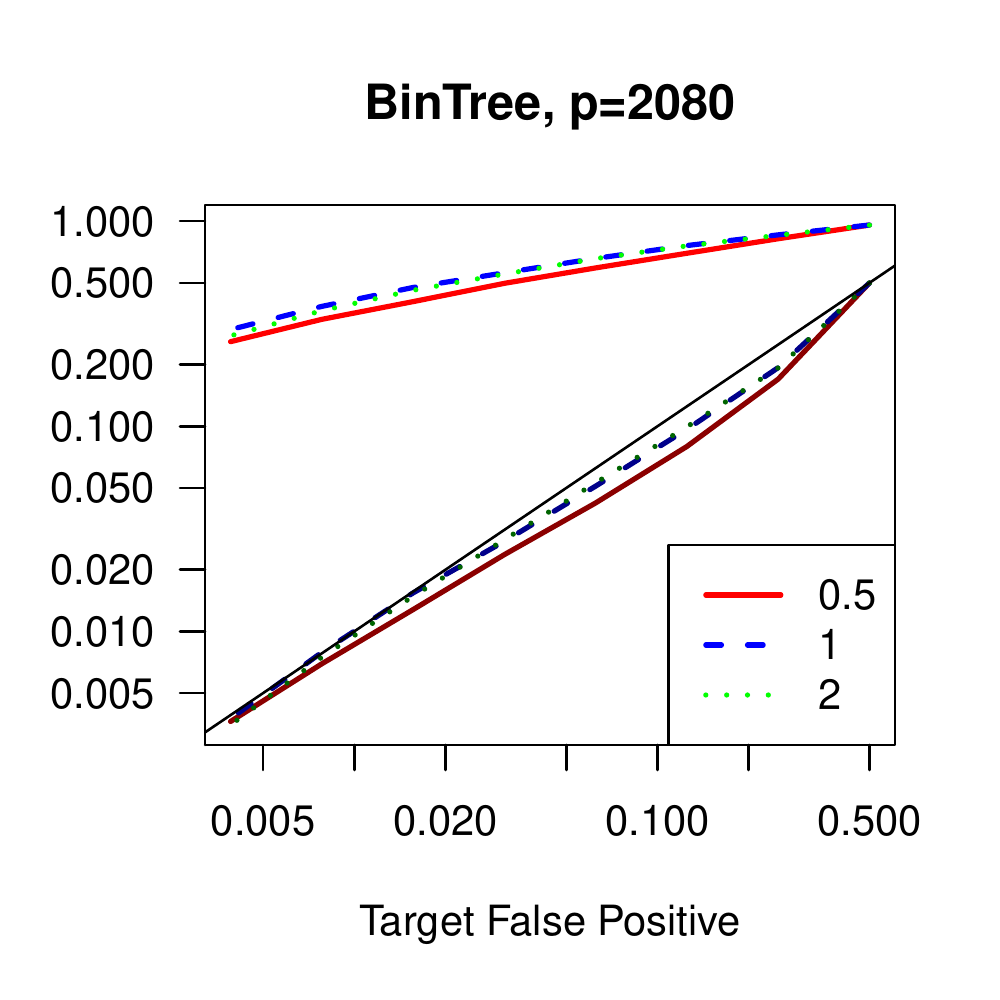}
  \includegraphics[width=0.32\textwidth]{\PICDIR/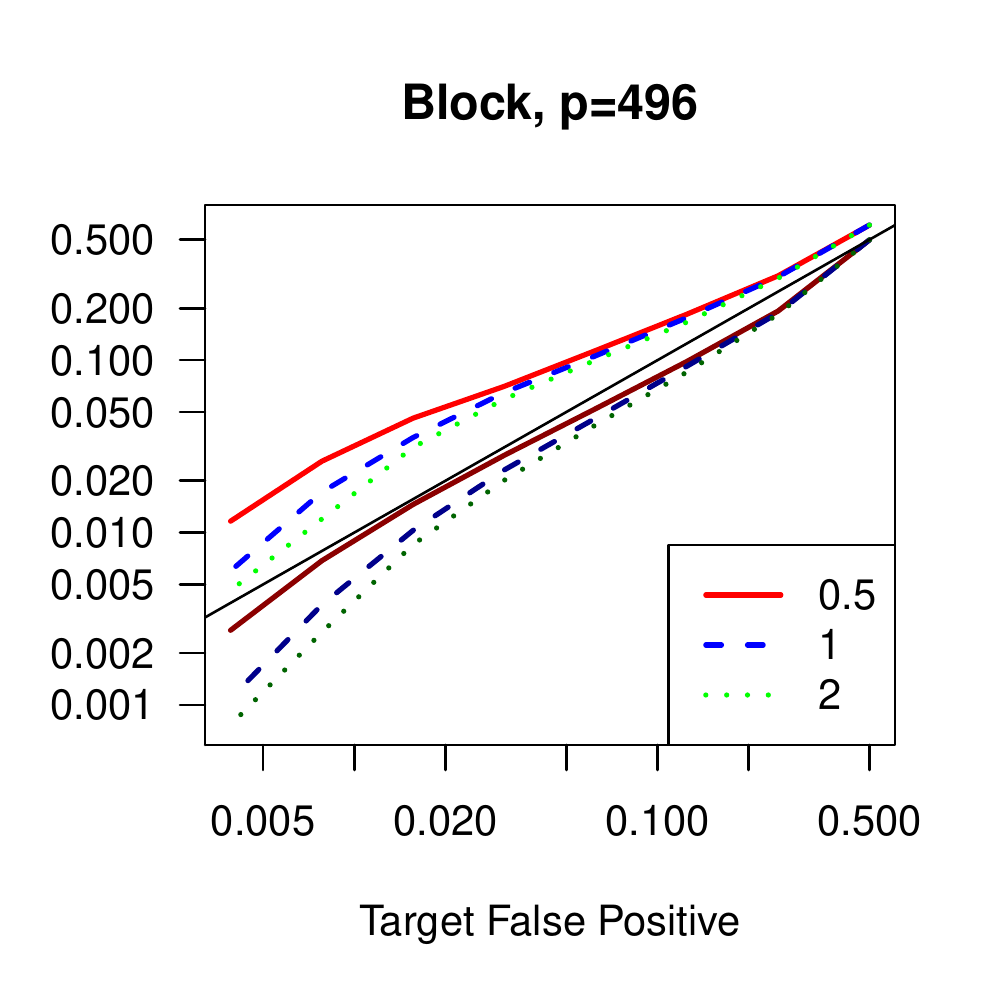}
  \includegraphics[width=0.32\textwidth]{\PICDIR/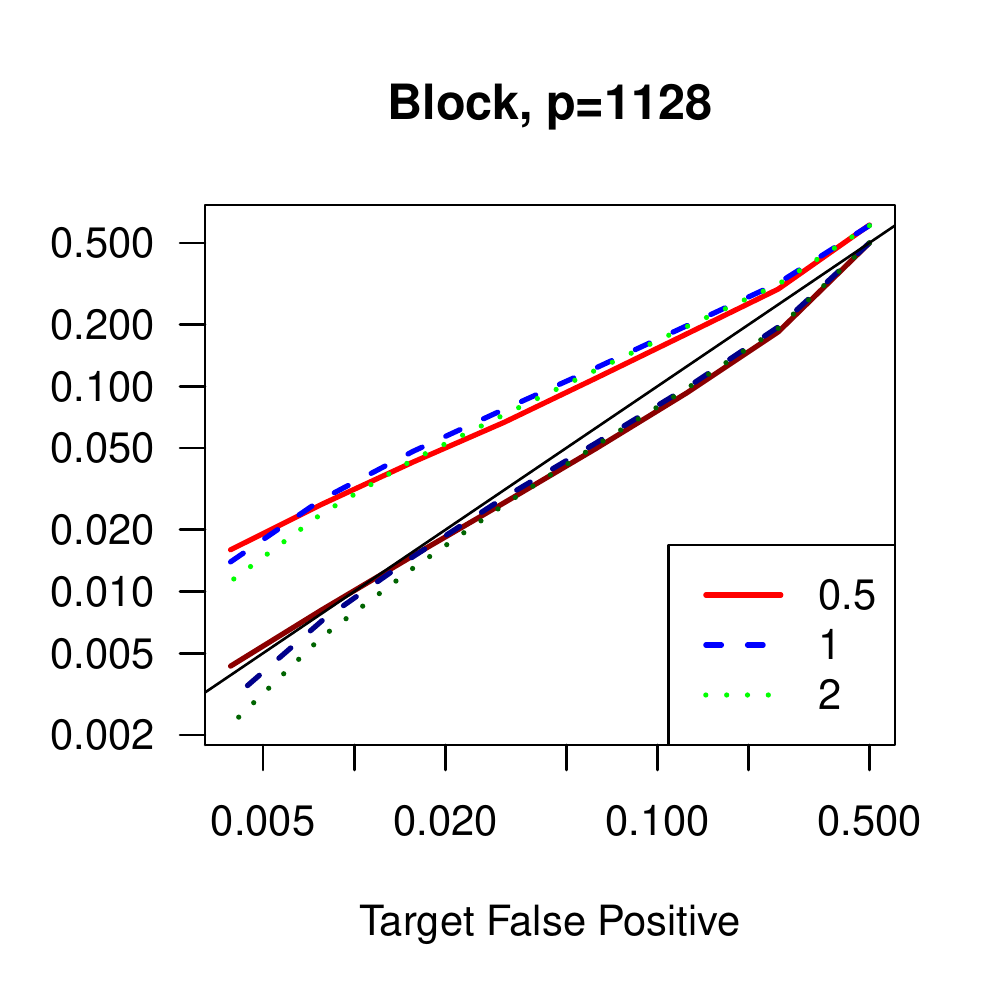}
  \includegraphics[width=0.32\textwidth]{\PICDIR/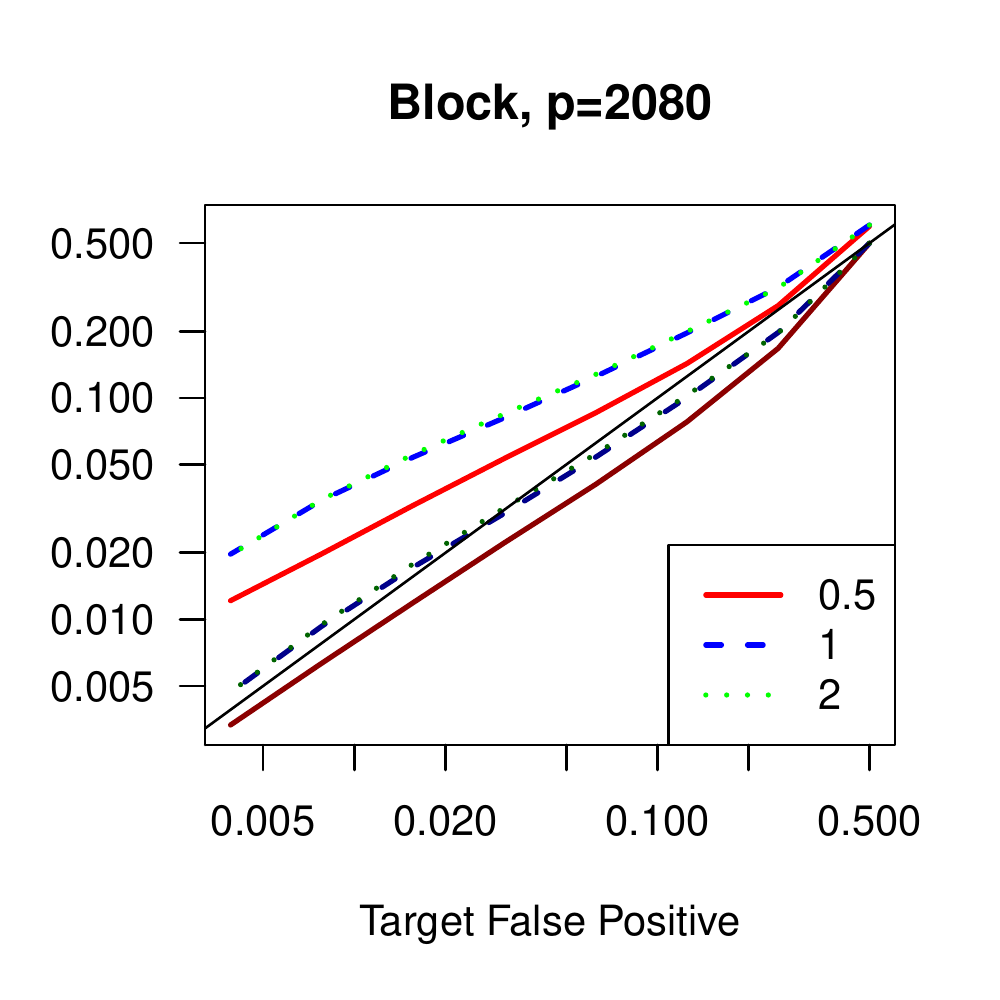}
  \end{center}
  \capt{
    \label{fig:fptpPlot}
    The achieved true and false positives plotted against the
    target false positive rate displayed on the log-log scale
    for multivariate Gaussian data.
    The rows from top to bottom correspond to the tridiagonal, 
    binary tree, and block diagonal matrices.  The columns from
    left to right correspond to dimensions 496, 1128, and 2080.
  }
\end{figure}

\begin{figure}
	\begin{center}
		\includegraphics[width=0.32\textwidth]{\PICDIR/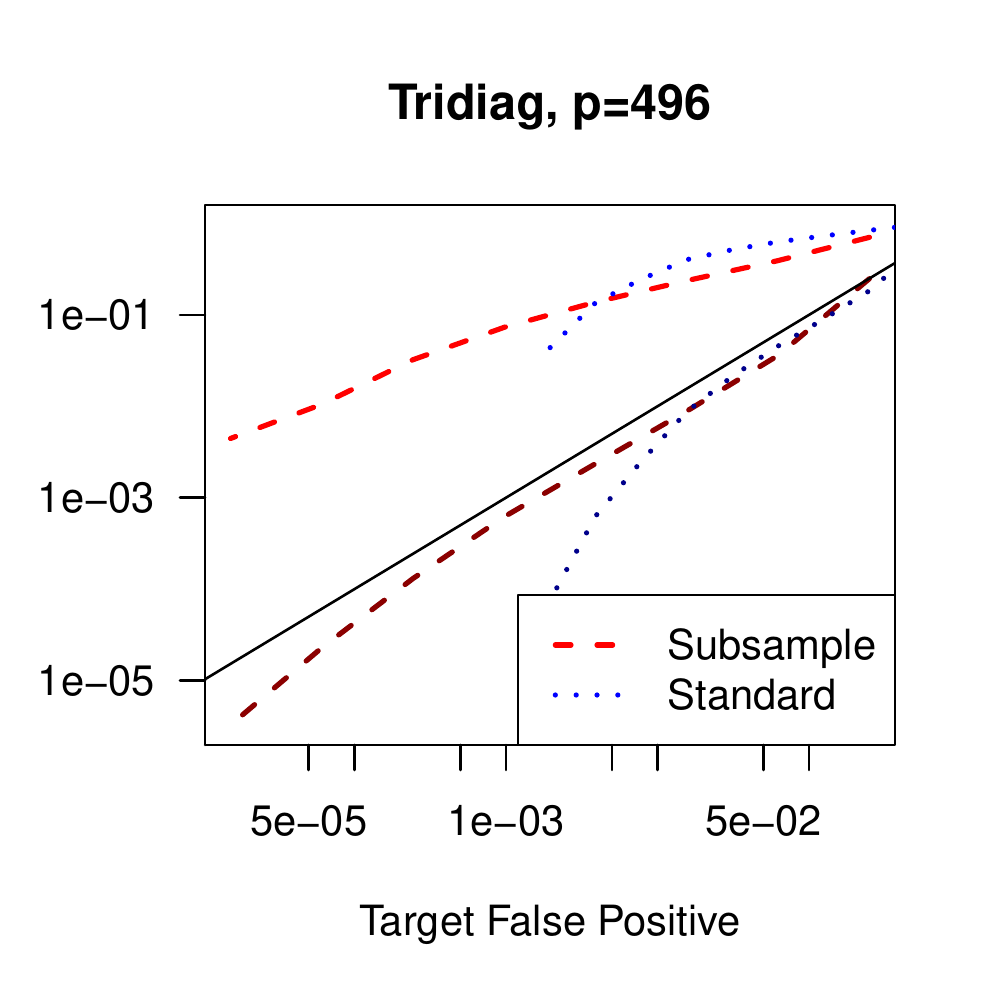}
		\includegraphics[width=0.32\textwidth]{\PICDIR/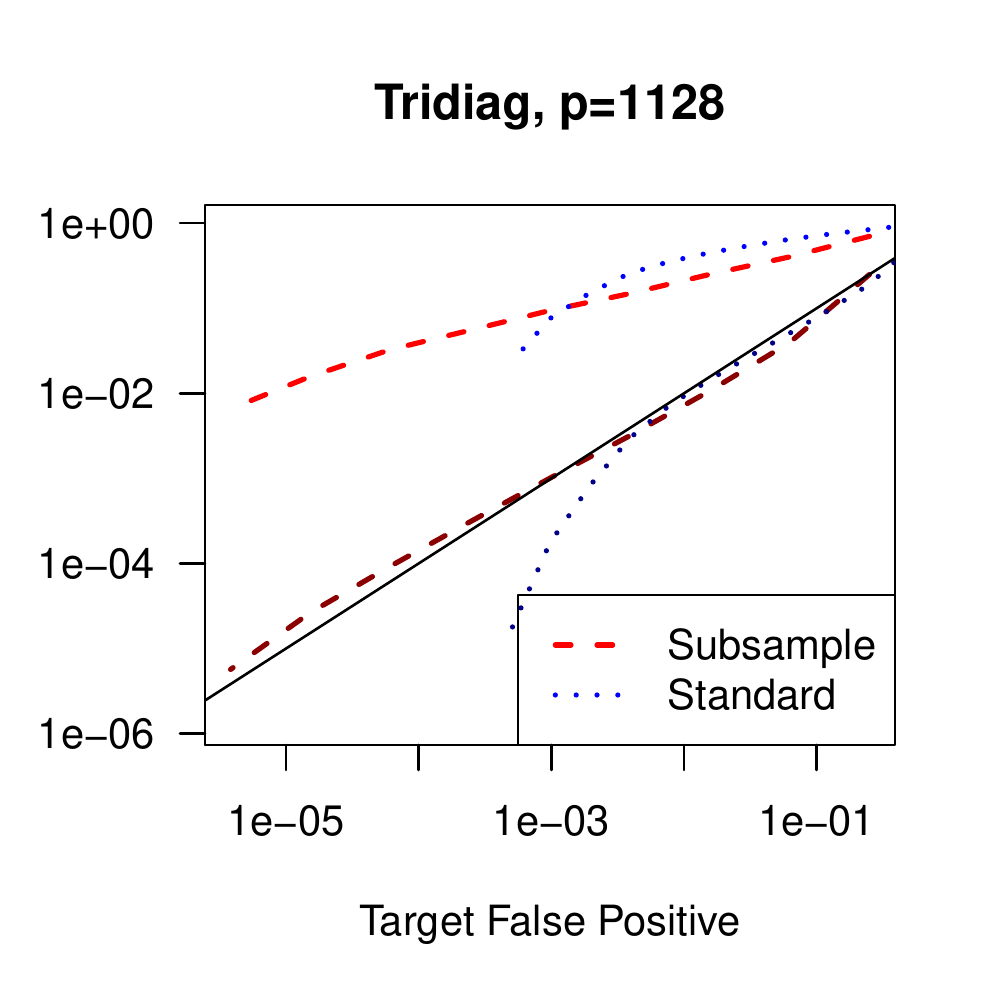}
		\includegraphics[width=0.32\textwidth]{\PICDIR/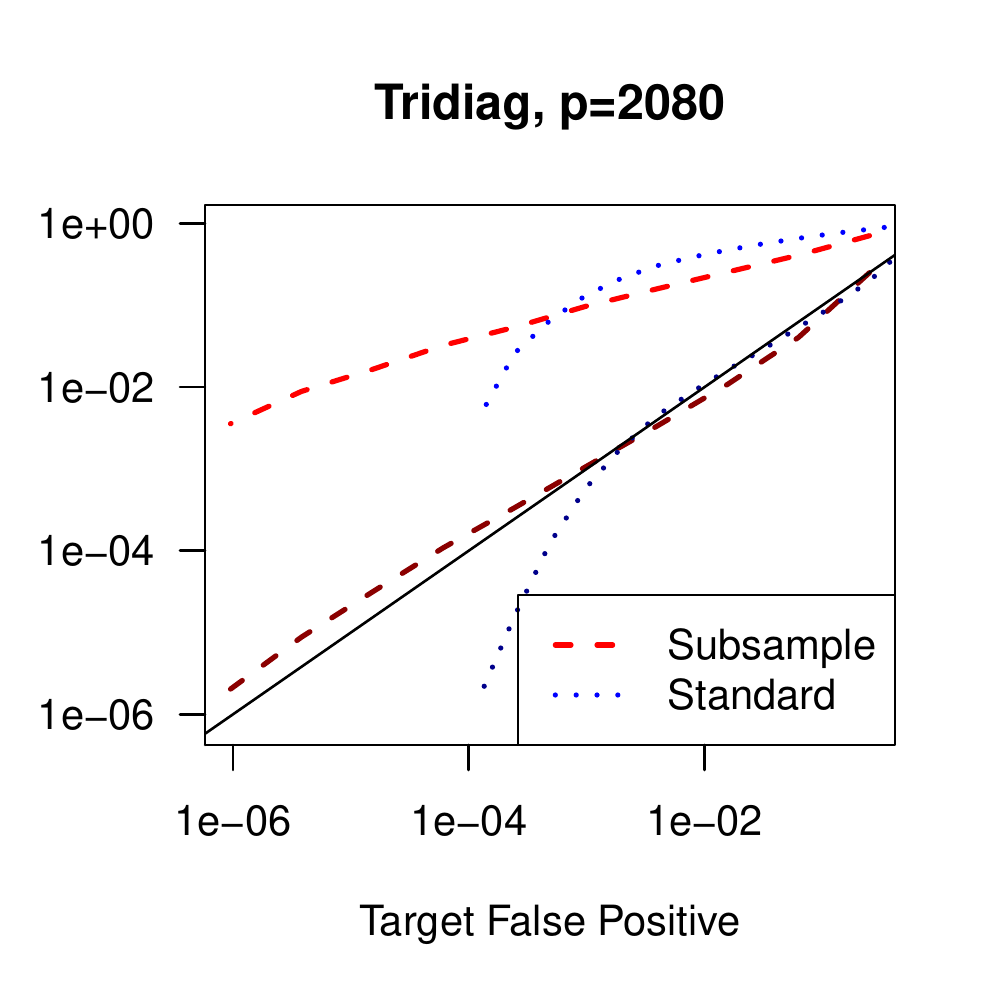}
		\includegraphics[width=0.32\textwidth]{\PICDIR/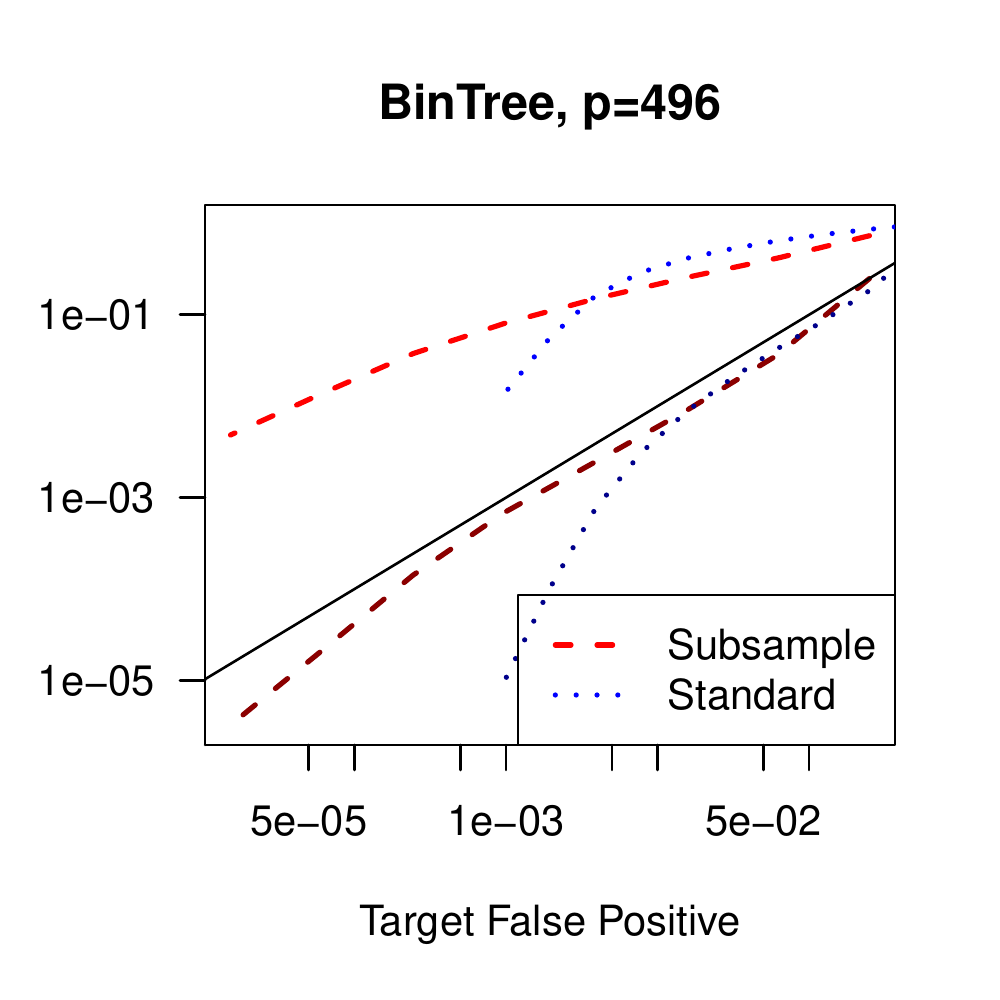}
		\includegraphics[width=0.32\textwidth]{\PICDIR/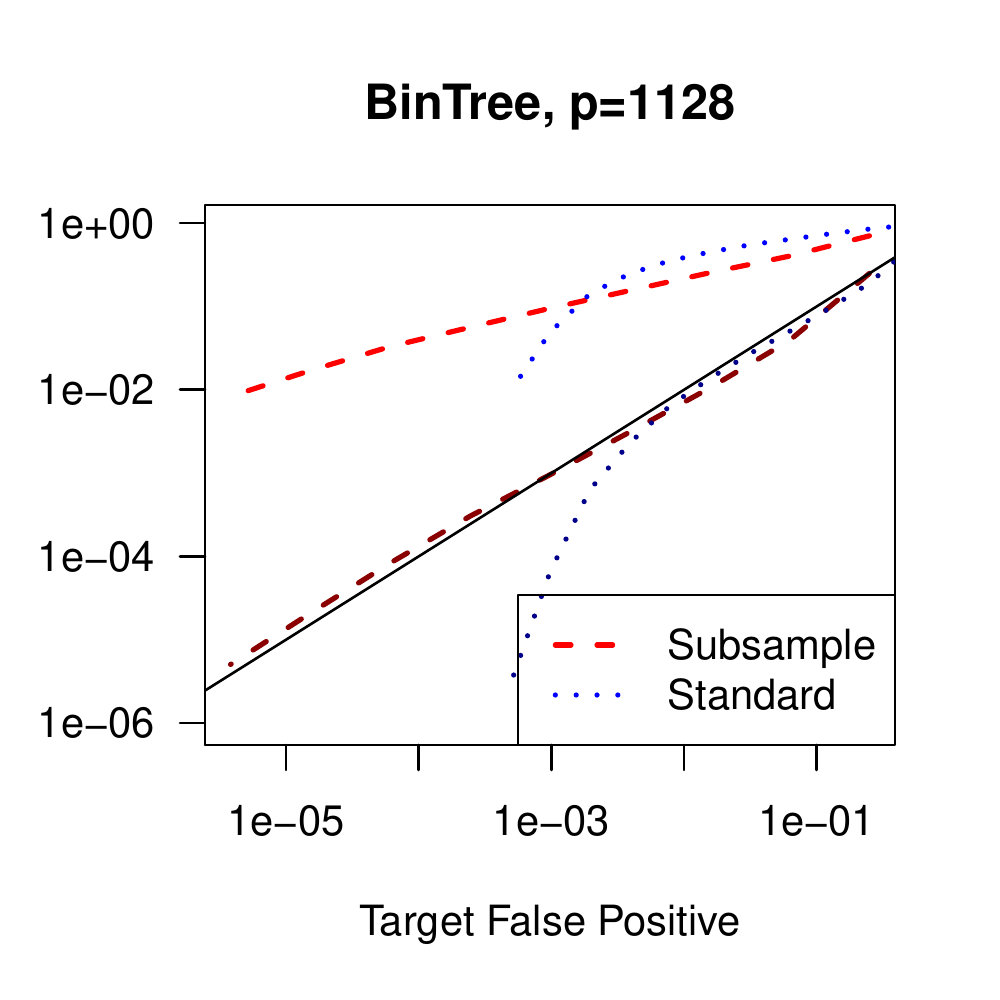}
		\includegraphics[width=0.32\textwidth]{\PICDIR/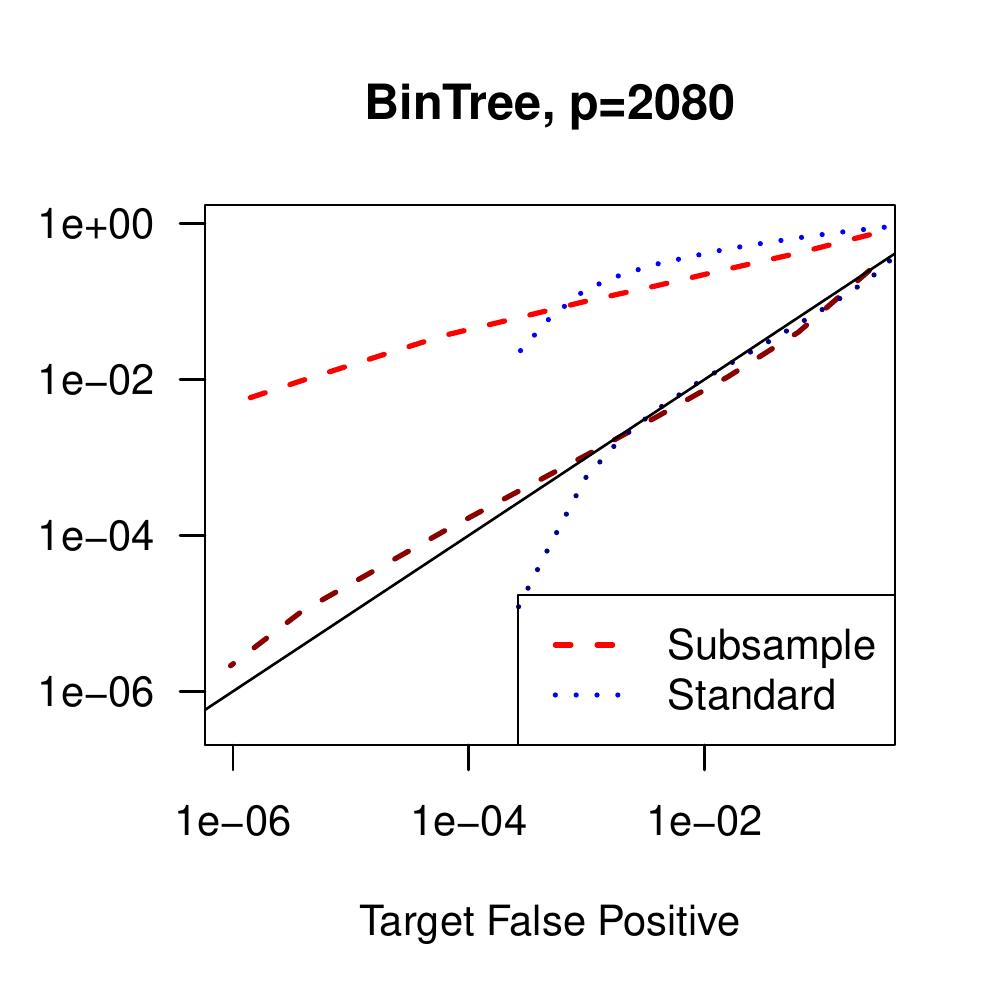}
	\end{center}
	\capt{
		\label{fig:subs}
		The achieved true and false positives plotted against the
		target false positive rate displayed on the log-log scale
		as in Figure~\ref{fig:fptpPlot}. The red lines are for the
		standard approach while blue lines are for subsampling by
		dividing the set in half.
	}
\end{figure}

\begin{table}
  \centering
  \fbox{
  \begin{tabular}{c|rrr||rrr}
  	& \multicolumn{6}{c}{Cross Validated GLASSO} \\
    & \multicolumn{3}{c}{False Pos} & \multicolumn{3}{c}{True Pos}\\\hline
    $p=$ &      496 & 1128 & 2080 & 496 & 1128 & 2080 \\
    TriDiag &  9.53 & 5.27 & 3.15 & 48.2& 36.3 & 29.9 \\
    BinTree &  9.71 & 5.22 & 3.15 & 47.3& 36.9 & 29.8 \\
    BlkDiag &  10.11& 5.37 & 3.20 & 14.1&  8.6 &  5.5 \\\hline\hline
      	& \multicolumn{6}{c}{Cross Validated ACLIME} \\
    & \multicolumn{3}{c}{False Pos} & \multicolumn{3}{c}{True Pos}\\\hline
    $p=$ &      496 & 1128 & 2080 & 496  & 1128 & 2080 \\
    TriDiag &  1.27 & 0.20 & 0.12 & 1.78 & 0.42 & 0.17 \\
    BinTree &  1.17 & 0.23 & 0.12 & 1.92 & 0.29 & 0.03 \\
    BlkDiag &  0.96 & 0.16 & 0.11 & 1.05 & 0.16 & 0.05 \\
  \end{tabular}}
  \capt{
  	\label{tab:glsclm}
  	The true and false positive percentages out of 100
        achieved by cross validated
  	graphical lasso (top) and ACLIME (bottom) estimators with respect to 
  	the operator norm distance.
  }
\end{table}

\subsection{Geonomics data}

We apply our methodology to the dataset of gene expressions
for a small round blue cell tumours mircoarray experiment from 
\cite{KHAN2001}, which is also analyzed in other works 
\citep{ROTHMAN2009,CAILUI2011}.
The data set consists of
a training set of 64 vectors containing 2308 gene expressions.
The data contains four types of tumours.
Considering this $X\in\real^{64\times2308}$, we construct
error controlled precision matrix estimators using the 
algorithm from Section~\ref{sec:algorithm} choosing the
debiased estimator of \cite{JANKOVA2015} as the initial 
estimator and for graphical lasso penalization parameter
$\lmb=0.5,1,2$.

In Figure~\ref{fig:khanNZ}, we plot the cardinality of the 
support of the estimators for the three penalization 
parameters as the false positive rate ranges from 
$2^{-1}$ to $2^{-10}$.  For $\lmb=0.5$, the method
plateaus quickly returning a small support of 74 nonzero
off-diagonal entries.  The other two lines continue to decay
towards zero off-diagonal entries.

In Figure~\ref{fig:khanRows}, we see kernel density plots of
the number of non-zero entries per row aggregated over all 
2308 rows.  In these plots, 
$\lmb = 0.5,1,2$ and $\alpha=0.5,0.001$ are considered.
Comparing the top and bottom row of plots, we see that 
there are many rows with approx 500 nonzero entries 
when $\alpha=0.5$, which quickly drops to fewer than 
10 nonzero entries per row once $\alpha=0.001$.

\begin{figure}
  \begin{center}
  \includegraphics[width=0.5\textwidth]{\PICDIR/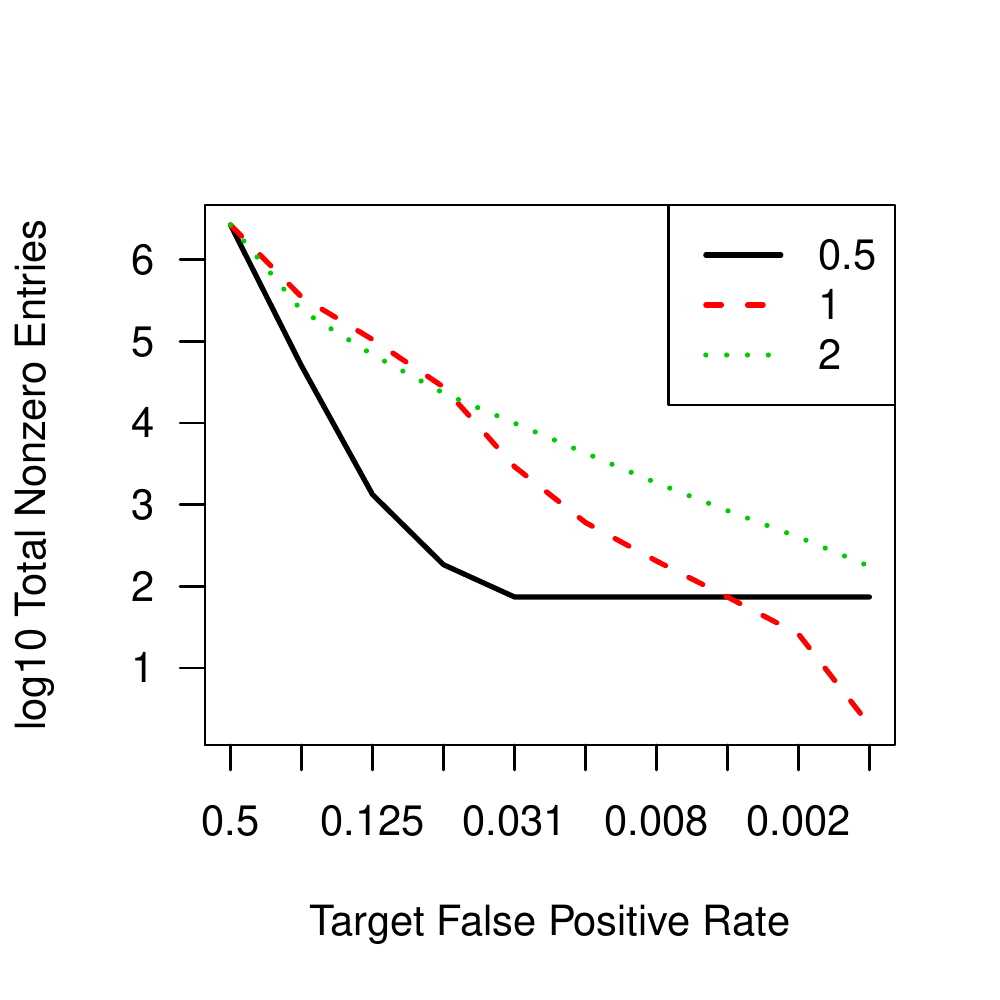}
  \end{center}
  \capt{
    \label{fig:khanNZ}
    The $\log_{10}$ number of non-zero entries in the $2308\times2308$
    precision estimator for graphical lasso penalizations of $0.5,1,2$
    and for false positive rates from $2^{-1}$ to $2^{-10}$.
  }
\end{figure}
\begin{figure}
  \begin{center}
  \includegraphics[width=0.32\textwidth]{\PICDIR/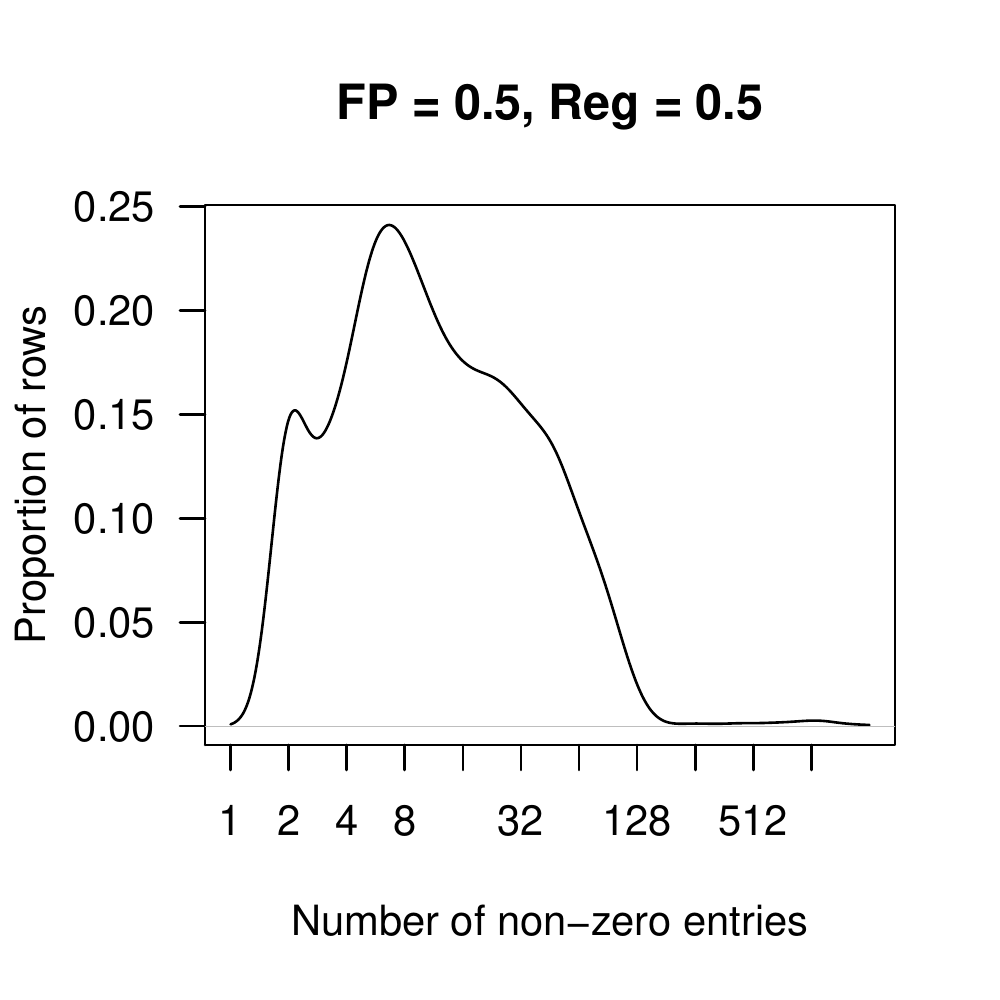}
  \includegraphics[width=0.32\textwidth]{\PICDIR/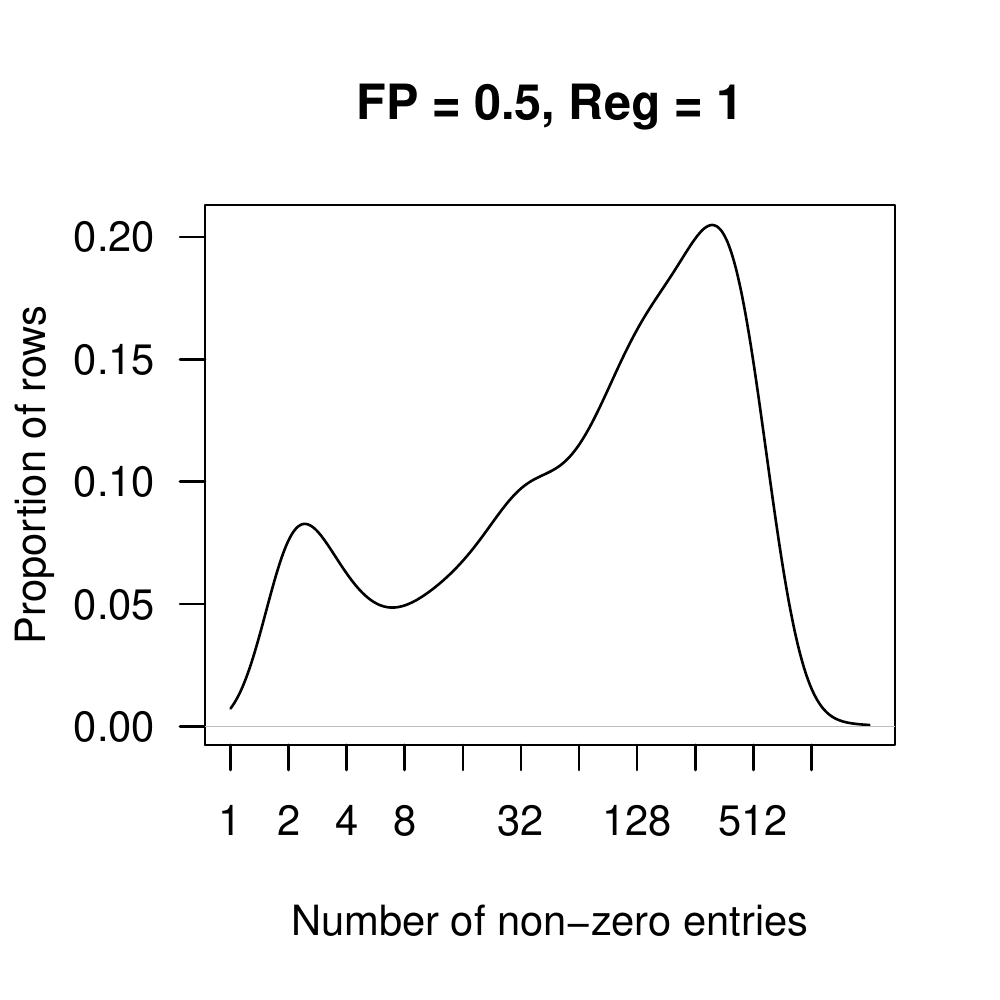}
  \includegraphics[width=0.32\textwidth]{\PICDIR/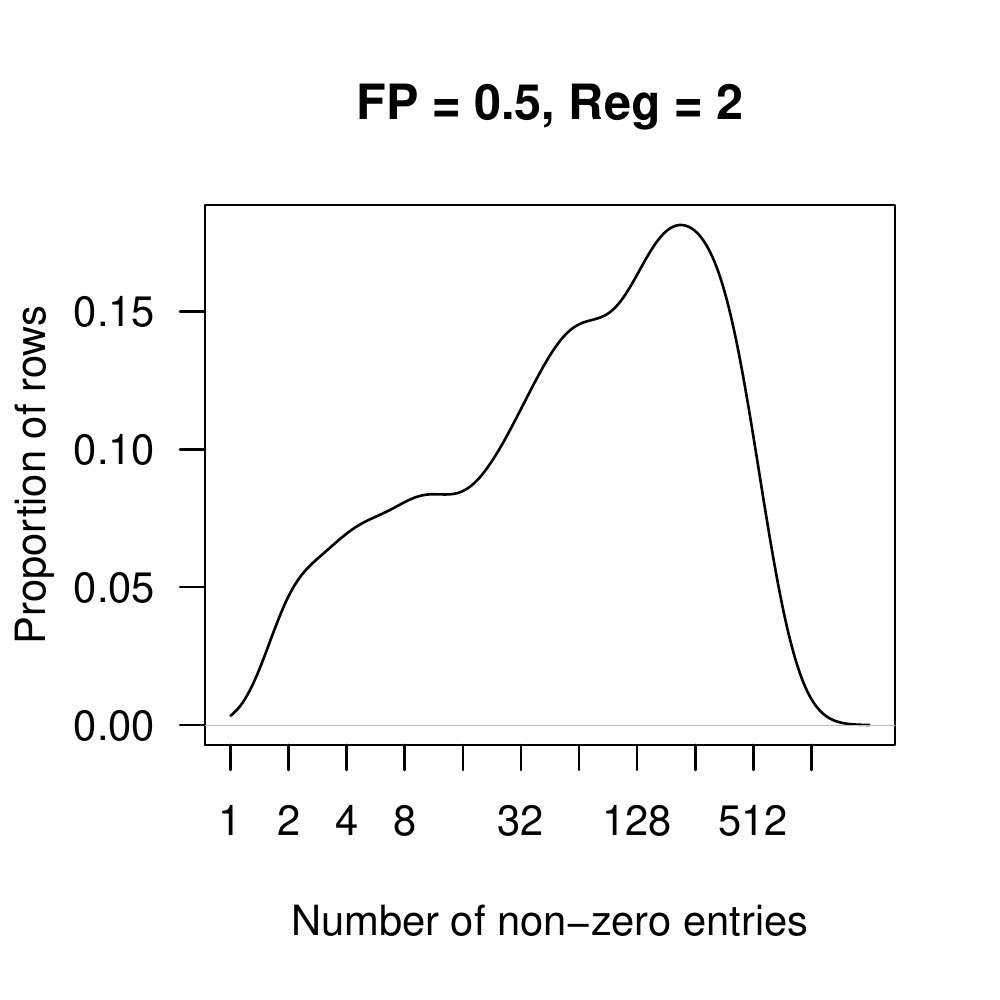}
  \includegraphics[width=0.32\textwidth]{\PICDIR/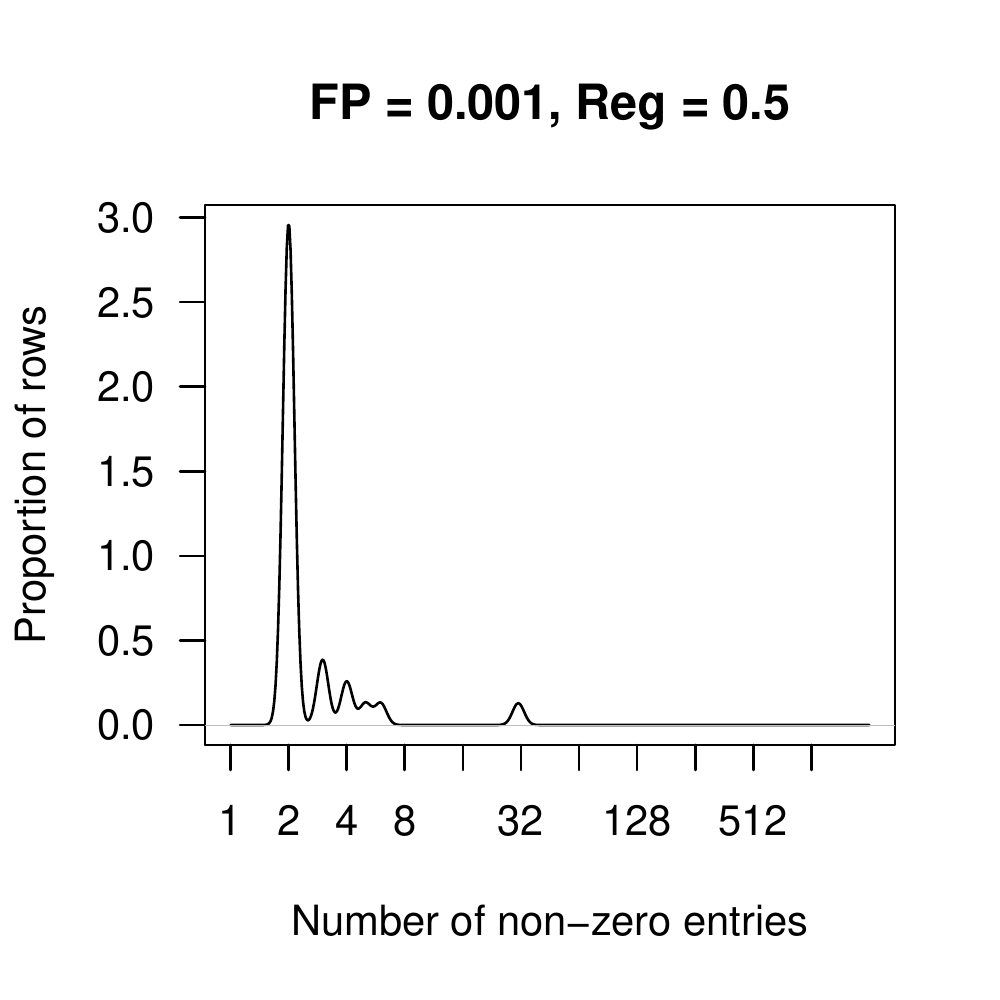}
  \includegraphics[width=0.32\textwidth]{\PICDIR/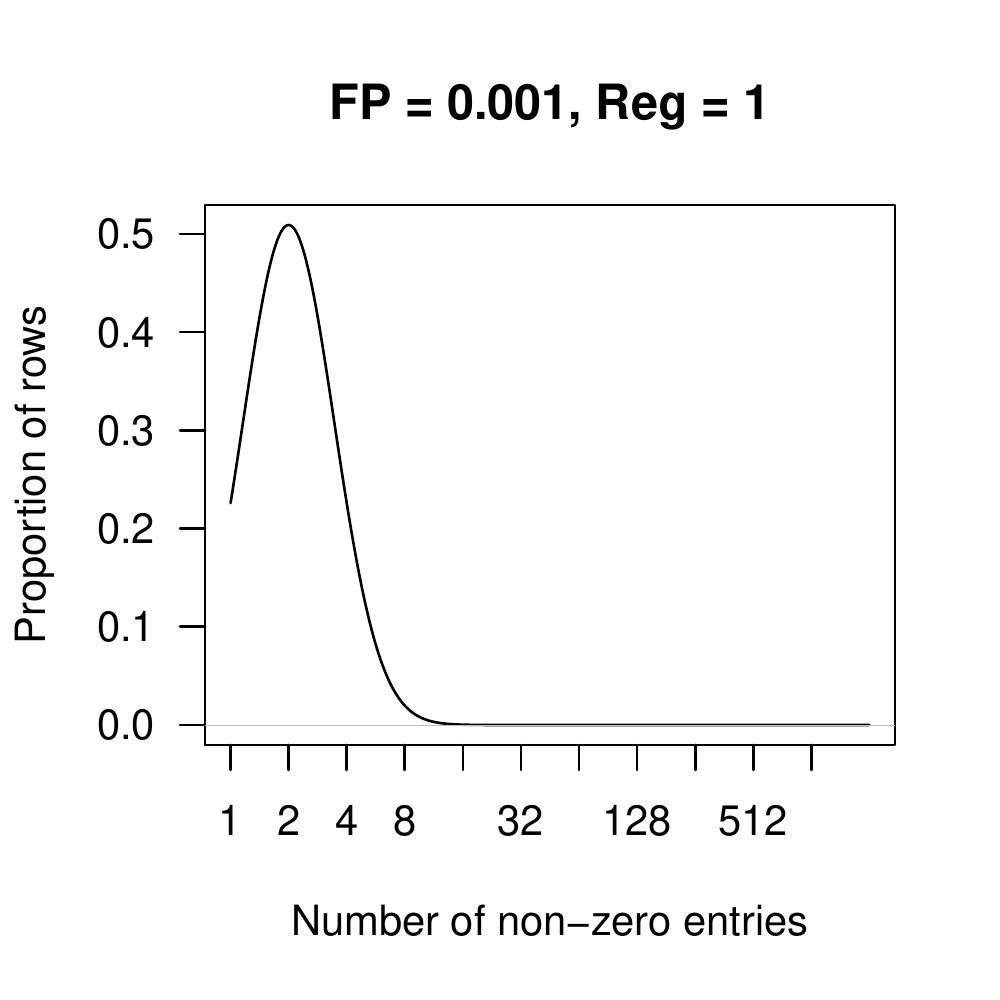}
  \includegraphics[width=0.32\textwidth]{\PICDIR/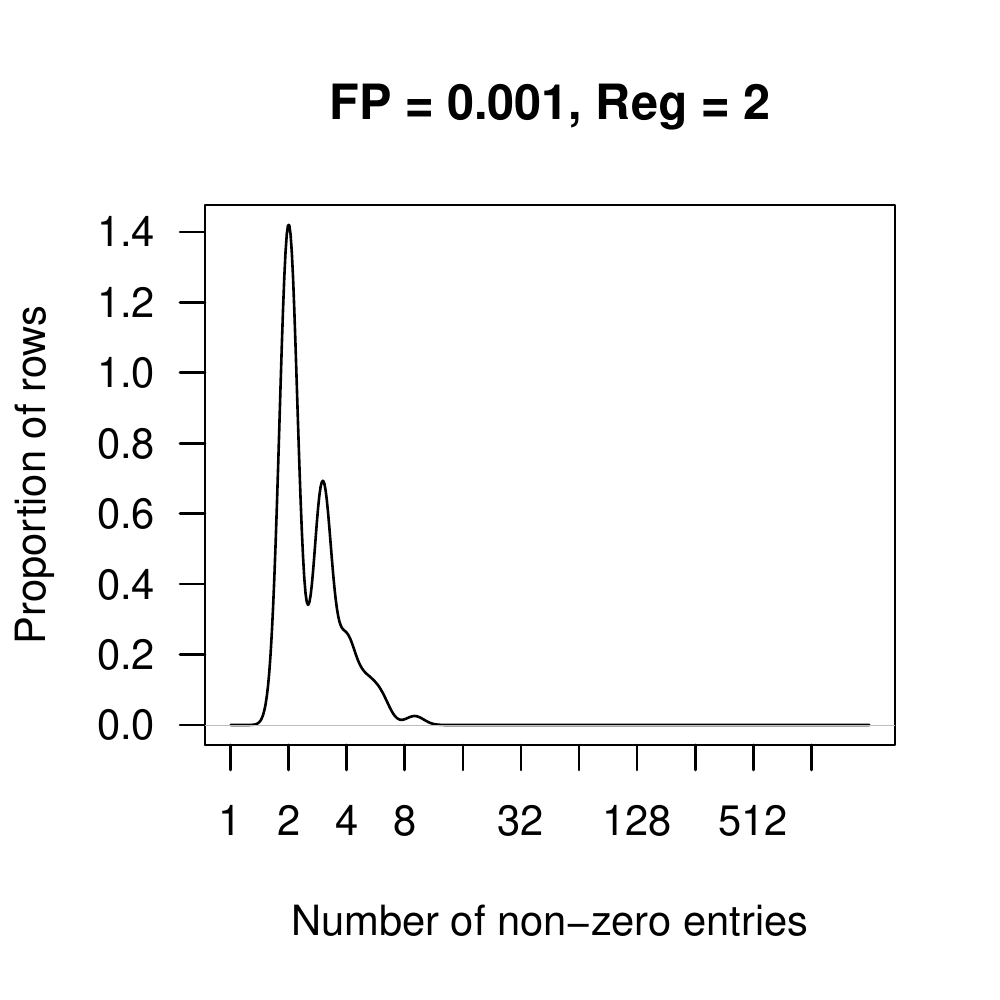}
  \end{center}
  \capt{
    \label{fig:khanRows}
    Density plots of the number of nonzero entries per row
    for precision estimators with, from left to right, 
    graphical lasso penalizations of $0.5,1,2$ and with 
    false positive rates 0.5, top row, and 0.001, bottom row.
  }
\end{figure}

\section{Supplementary Material}

Supplementary material available at Biometrika online includes
proofs for the main results and auxiliary results. 
Receiver operating characteristic curves for the data in 
Figure~\ref{fig:fptpPlot} are included, and
additional simulations for exponential data and the method from 
\cite{JANKOVA2015} are also included.


\bibliographystyle{biometrika}
\bibliography{\BIBDIR/kasharticle,\BIBDIR/kashbook,\BIBDIR/kashself,\BIBDIR/kashpack}

\section{Proofs}

The theory behind the methodology in this article comes from 
redoing the results from random matrix theory, which generally
achieves an operator norm of $O(p^{1/2})$,  under the assumption
that only $\alpha\in(0,1)$ of the entries are non-zero.  This
gives similar results but with an operator norm 
of $O(\alpha^{1/2}p^{1/2})$.

We first prove a lazy version of Hoeffding's concentration
inequality.
\begin{lemma}
  \label{lem:lazyhoef}
  For some $M>0$,
  let $Z_1,\ldots,Z_p\in[-M,M]$ be independent mean zero random 
  variables not necessarily equally distributed,
  and let $b_1,\ldots,b_p\distiid\distBern{\alpha}$ also independent
  of the $Z_i$.  Then,
  $$
    \mathrm{P}\left({
      \abs{\sum_{i=1}^p b_iZ_i} \ge tM(\alpha p)^{1/2}
    }\right) \le 2\ee^{-t^2/4}.
  $$
\end{lemma} 
\begin{proof}
  The proof is more or less the same as that of the standard
  Hoeffding inequality.  Without loss of generality, 
  let $M=1$.
  For any $i=1,\ldots,p$, we can bound the variance of $b_iZ_i$
  by 
  $$
    \var{b_iZ_i} = \xval{b_i^2}\xval{Z_i^2} \le \alpha.
  $$
  Thus, we have a bound on the cumulant generating function
  of $b_iZ_i$ that
  $
    \log\xv\ee^{\lmb b_iZ_i} \le \alpha\lmb^2
  $ and 
  $
    \log\xv\ee^{\lmb \sum_{i=1}^p b_iZ_i} \le \alpha p\lmb^2.
  $
  Therefore, via Chernoff's inequality, we have for any $\lmb>0$ that
  $$
    \mathrm{P}\left({\sum_{i=1}^p b_iZ_i \ge t}\right)
    \le
    \mathrm{P}\left({
      \ee^{\lmb\sum_{i=1}^p b_iZ_i} \ge 
      \ee^{\lmb t}
    }\right)
    \le \ee^{-\lmb t} \xv \exp\left(
      \lmb\sum_{i=1}^p b_iZ_i
    \right)
    \le \exp\left(
      -\lmb t + p\alpha\lmb^2
    \right)
  $$
  Minimizing the righthand side 
  over $\lmb$ gives $\lmb = t/(2p\alpha)$ and finally
  that 
  $
    \mathrm{P}\left({\sum_{i=1}^p b_iZ_i \ge t}\right) \le
    \ee^{ -t^2/(4\alpha p) }.
  $
  Running the proof for $\sum_{i=1}^p b_iZ_i \le -t$
  gives the reverse inequality.  Combining those gives
  the desired result.
  Lastly, to adjust for $M\ne1$, replace $Z_i$ with $Z_i/M$
  in the above result.
\end{proof}

Using the lazy version of Hoeffding's inequality, we can 
get concentration of the operator norm for a 
sparse random matrix using a modification of the 
standard $\veps$-net argument as presented in \cite{TAO2012} Section~2.3.

\begin{theorem}
  \label{thm:matConc}
  Let $A\in\real^{p\times p}$ be a symmetric random matrix
  with entries $A_{i,i}=0$, $\abs{A_{i,j}}<1$, and 
  $A_{i,j}=A_{j,i}$.  Let $B\in\{0,1\}^{p\times p}$ be 
  a symmetric Bernoulli random matrix with iid entries
  such that $\mathrm{P}({B_{i,j}=1})=\alpha\in[0,1]$.  If the lower 
  triangular entries of $A$ are independent then,
  $$
    \xv\norm{A\circ B}_\infty = O\{ (\alpha p)^{1/2} \}
  $$
  where $A\circ B$ denotes the entrywise or Hadamard product
  of $A$ and $B$.
\end{theorem}
\begin{proof}
  For now, we consider $A$ and $B$ as iid ensembles---i.e. remove the 
  $A_{i,j}=A_{j,i}$ and $B_{i,j}=B_{j,i}$ condition---and adjust for the
  symmetry at the end of the proof.
  Let $\Pi = A\circ B$.
  For any vector $v\in\real^p$ such that $\norm{v}_2=1$,
  we have that 
  $
    \norm{\Pi v}_2 = 
    (\sum_{i=1}^p(\sum_{j=1}^pA_{i,j}B_{i,j}v_j)^2)^{1/2}.
  $
  Therefore, from the proof of Lemma~\ref{lem:lazyhoef}, independence 
  of entries, and $\norm{v}_2=1$, 
  $$
    \log\xv\ee^{\lmb\norm{\Pi v}_2} =  
    \sum_{i=1}^p\log\xv\ee^{\lmb \sum_{j=1}^p A_{i,j}B_{i,j}v_j} \le 
    \alpha p\lmb^2.
  $$
  Thus, for any arbitrary unit vector $v$, we can 
  apply Lemma~\ref{lem:lazyhoef} to get
  $$
    \mathrm{P}\left({
      \norm{\Pi v}_2 > t(\alpha p)^{1/2}
    }\right)
    \le 2\ee^{-t^2/4}.
  $$
  To extend to $\norm{\Pi}_\infty$, we construct a maximal 
  $(1/2)$-net of the unit sphere $S = \{v\,:\,\norm{v}_2=1\}$.
  That is, let $V$ be the maximal set of points in $S$ such that 
  for any $u,v\in V$, $\norm{u-v}_2\ge0.5$.  Let $v^*\in S$ be
  the vector such that $\norm{\Pi}_\infty = \norm{\Pi v^*}_2$,
  which exists via compactness.
  Thus, there exists a $v_0\in V$ such that $\norm{v^*-v_0}_2<0.5$
  as otherwise, the set $V$ would not be maximal.
  Furthermore, $\norm{\Pi(v^*-v_0)}_2\le\norm{\Pi}_\infty/2$
  and thus $\norm{\Pi v_0}_2\ge\norm{\Pi}_\infty/2$.
  Therefore, by the union bound,
  $$
    \mathrm{P}\left({
      \norm{\Pi}_\infty > t(\alpha p)^{1/2}
    }\right) \le
    \mathrm{P}\left({
      \bigcup_{v\in V}\left\{
      \norm{\Pi v}_2 > \frac{t}{2}(\alpha p)^{1/2}
      \right\}
    }\right) \le
    2\abs{V}\ee^{-t^2/16}.
  $$
  From Lemma~2.3.4 of \cite{TAO2012}, $\abs{V} = (2C)^p$ for
  some absolute constant $C>0$. Thus, for $t$ large enough,
  the righthand side becomes negligible.
  
  Now, considering $A$ and $B$ symmetric with zero diagonal
  as in the theorem statement, we have that 
  $
    \norm{\Pi}_\infty \le \norm{\Pi_\text{low}}_\infty + 
    \norm{\Pi_\text{up}}_\infty
  $
  for $\Pi_\text{low}$ and $\Pi_\text{up}$ the strict lower and upper 
  triangular portions of $\Pi$.  Hence,
  $$
  \mathrm{P}\left({
  	\norm{\Pi}_\infty > t(\alpha p)^{1/2}
  }\right) \le
  \mathrm{P}\left({
    \norm{\Pi_\text{low}}_\infty > t(\alpha p)^{1/2}/2
  }\right) +
  \mathrm{P}\left({
  	\norm{\Pi_\text{up}}_\infty > t(\alpha p)^{1/2}/2
  }\right)
  \le
  4\abs{V}\ee^{-t^2/32}.
  $$
\end{proof}
\begin{remark}
  The constants in the above proof are not necessarily
  optimal, but sufficient to justify our approach to controlled 
  precision matrix estimation.
\end{remark}

\begin{proof}[of Theorem~1]
  For (a), $\hat{\Omega}_{{s}}-\Omega$
  has zero diagonal and off-diagonal entries bounded by two.
  Further, if 
  $\Omega\in\mathcal{U}(\kappa,\delta)$, then 
  $\norm{\Omega}_\infty \le \max_i\sum_j \abs{\Omega_{i,j}} \le \kappa$.
  Hence, if $\kappa = O(p^\nu)$, then so is $\norm{\Omega}_\infty$.
  For the bias, let $B_{{s}}$ be a Bernoulli random 
  matrix with probability $\gamma^{-s}$ that entry $B_{i,j}$ is 1.
  Then,
  $$ 
    \text{bias}(\hat{\Omega}_{{s}})
    = \xv(\hat{\Omega}_{{s}})-\Omega
    = \xv( B_{{s}}\circ\hat{\Omega})-\Omega
    = \gamma^{-s}\text{bias}(\hat{\Omega}) + 
      (1-\gamma^{-s})\Omega = o(p^{1/2}).
  $$
  Thus, we have by assumption and by applying 
  theorem~\ref{thm:matConc} as well as the Bai-Yin 
  theorems in \cite{TAO2012} Section~2.3 that
  \begin{multline*}
    \xv\norm{\hat{\Omega}_{{s}}-\Omega}_\infty
    \le
    \xv\norm{
      \hat{\Omega}_{{s}}-\Omega-\text{bias}(\hat{\Omega}_{{s}})
    }_\infty + \norm{\text{bias}(\hat{\Omega}_{{s}})}_\infty
    \le\\\le 
    (2+o(1))p^{1/2}\gamma^{(s-1)/2} + o(p^{1/2})
  \end{multline*}
  and
  \begin{multline*}
    \xv\norm{\hat{\Omega}_{{s}}-\Omega}_\infty
    \ge
    \xv\norm{
      \hat{\Omega}_{{s}}-\Omega-\text{bias}(\hat{\Omega}_{{s}})
    }_\infty - \norm{\text{bias}(\hat{\Omega}_{{s}})}_\infty
    \ge\\\ge
    (2+o(1))p^{1/2}\gamma^{(s-1)/2} - o(p^{1/2}).
  \end{multline*}
  Therefore,
  $$
  \frac{ 
  	\xv\norm{\hat{\Omega}_{{s+1}}-\Omega}_\infty 
  }{
  	\xv\norm{\hat{\Omega}_{{s}}-\Omega}_\infty
  } = 
  \frac{ (2+o(1))p^{1/2}\gamma^{(s-1)/2} + o(p^{1/2}) }{ (2+o(1))p^{1/2}\gamma^{s/2} + o(p^{1/2}) }
  = \frac{1 + o(1)}{\gamma^{-1/2} + o(1)}.
  $$  
  
  For (b),
  let $K = \norm{\Omega-I_p}_\infty$ which is 
  $K\le \max_{i}\sum_{j=1}^p\abs{\Omega_{i,j}} \le \kappa = o(p^{1/2})$
  by the sparsity assumption.
  \begin{multline*}
  \frac{ 
  	\xv\norm{\hat{\Omega}_{{s+1}}-I_p}_\infty 
  }{
  	\xv\norm{\hat{\Omega}_{{s}}-I_p}_\infty
  } \le\frac{ 
  	\xv\norm{\hat{\Omega}_{{s+1}}-\Omega}_\infty
  	+ K 
  }{
  	\xv\norm{\hat{\Omega}_{{s}}-\Omega}_\infty
  	- K
  } 
  = \\ =
  \frac{ (2+o(1))\gamma^{(s-1)/2}  + o(1)+ Kp^{-1/2} 
  }{ 
    (2+o(1))\gamma^{s/2} + o(1) - Kp^{-1/2} 
  }
  = \frac{1 + o(1)}{\gamma^{-1/2} + o(1)}.
  \end{multline*}
  and similarly
  \begin{multline*}
  \frac{ 
  	\xv\norm{\hat{\Omega}_{{s+1}}-I_p}_\infty 
  }{
  	\xv\norm{\hat{\Omega}_{{s}}-I_p}_\infty
  } \ge\frac{ 
  	\xv\norm{\hat{\Omega}_{{s+1}}-\Omega}_\infty
  	- K 
  }{
  	\xv\norm{\hat{\Omega}_{{s}}-\Omega}_\infty
  	+ K
  } 
  = \\ =
  \frac{ (2+o(1))\gamma^{(s-1)/2}  + o(1)- Kp^{-1/2} 
  }{ 
    (2+o(1))\gamma^{s/2}  + o(1)+ Kp^{-1/2} 
  }
  = \frac{1 + o(1)}{\gamma^{-1/2} + o(1)}.
  \end{multline*}
\end{proof}

\section{Additional Simulations}

\subsection{Receiver operating characteristic curves}

In Figure~\ref{fig:rocs}, we consider plots of the observed
false positives against the observed true positives--that is,
Receiver operating characteristic curves.  
Similarly to the analysis of Figure~1 from the main article,
We see better performance for the tridiagonal and binary tree 
matrices than for the block diagonal matrix.  Also, the 
graphical lasso penalization parameter does not have a large effect 
on the methodology.  Though, $\lmb=1,2$ perform marginally better
than $\lmb=0.5$.  
\begin{figure}
	\begin{center}
		\includegraphics[width=0.32\textwidth]{\PICDIR/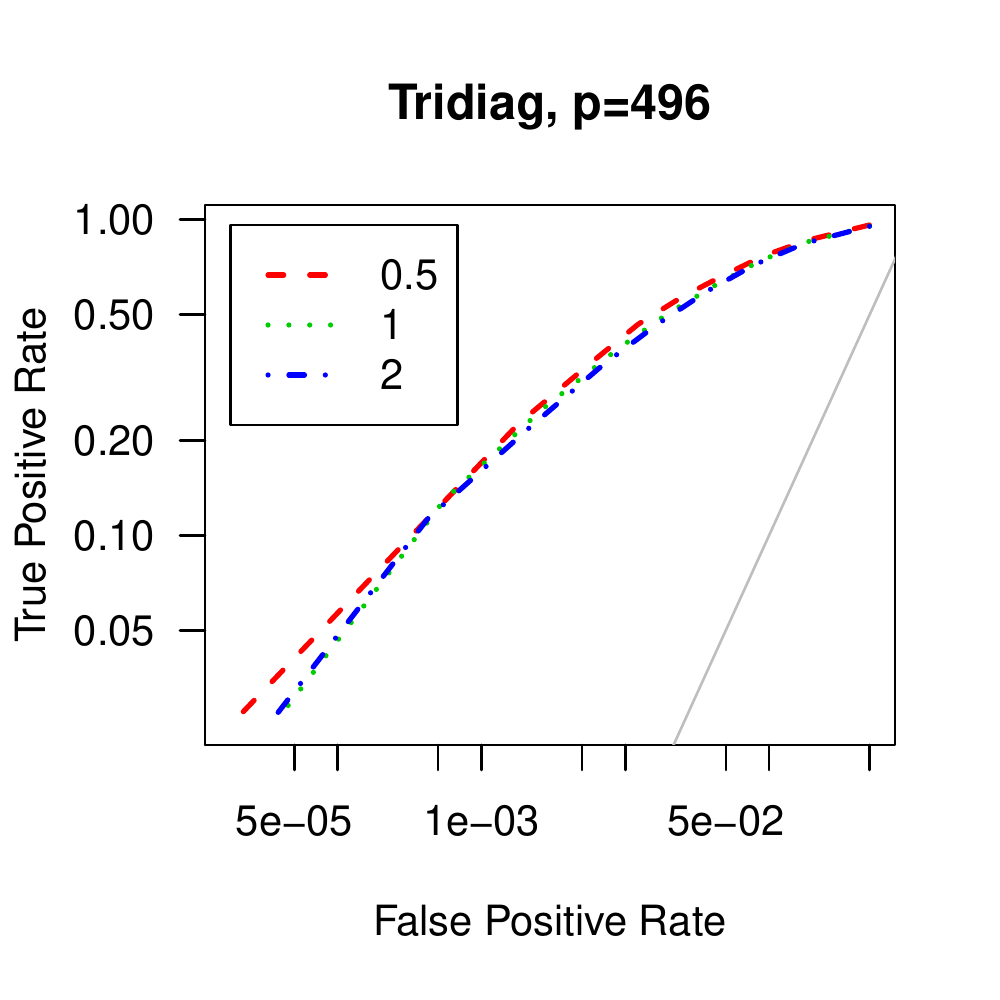}
		\includegraphics[width=0.32\textwidth]{\PICDIR/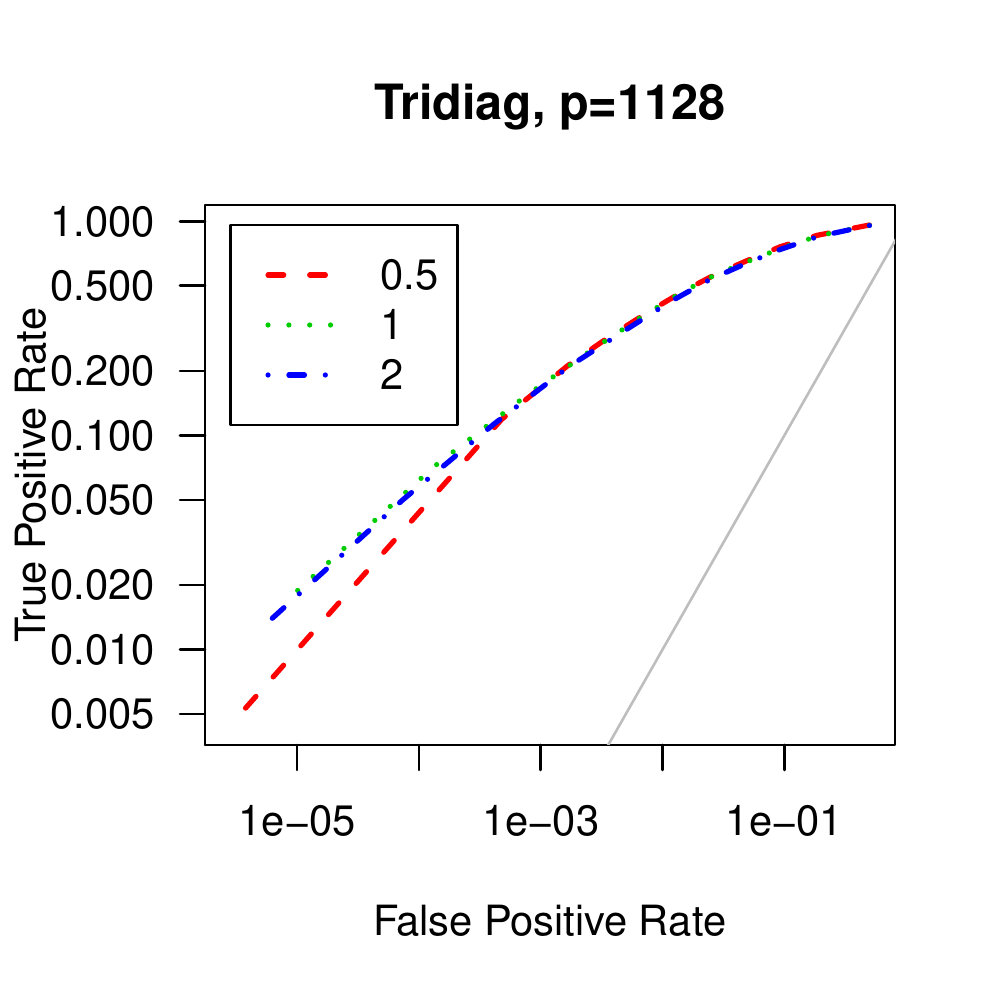}
		\includegraphics[width=0.32\textwidth]{\PICDIR/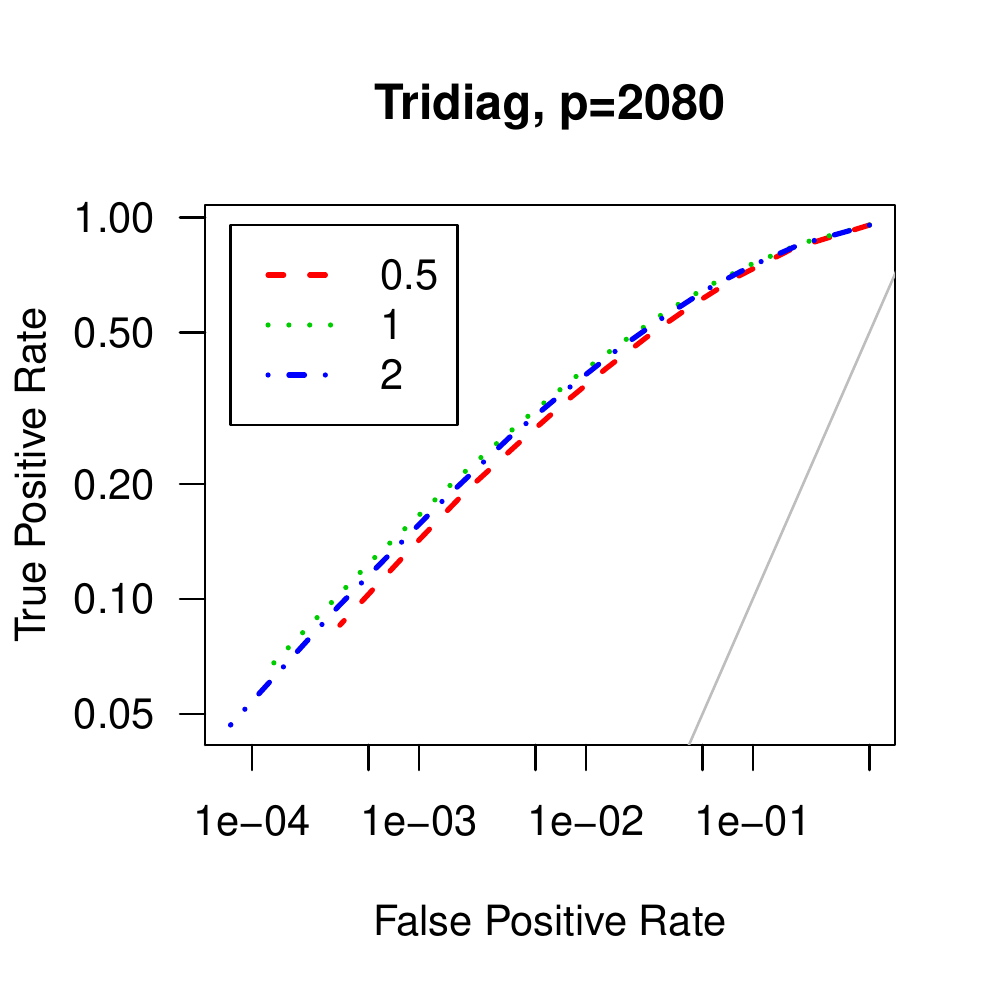}
		\includegraphics[width=0.32\textwidth]{\PICDIR/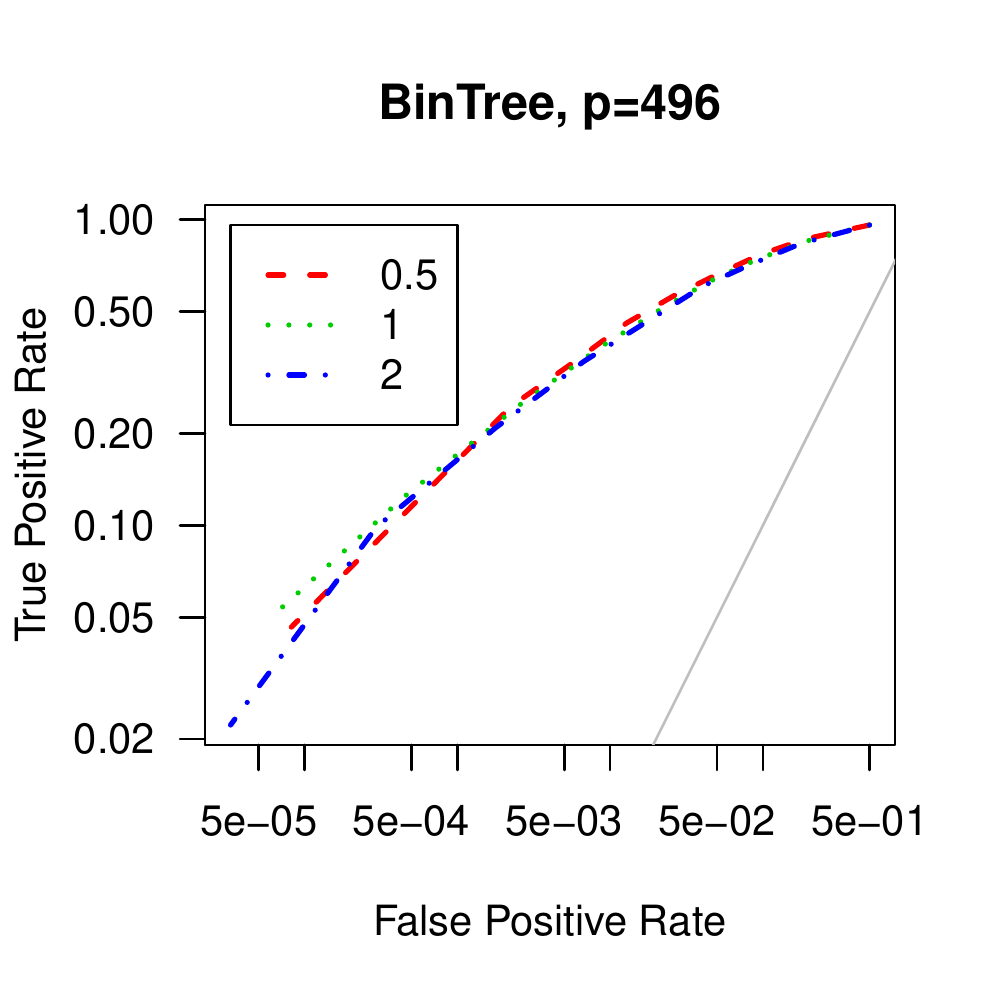}
		\includegraphics[width=0.32\textwidth]{\PICDIR/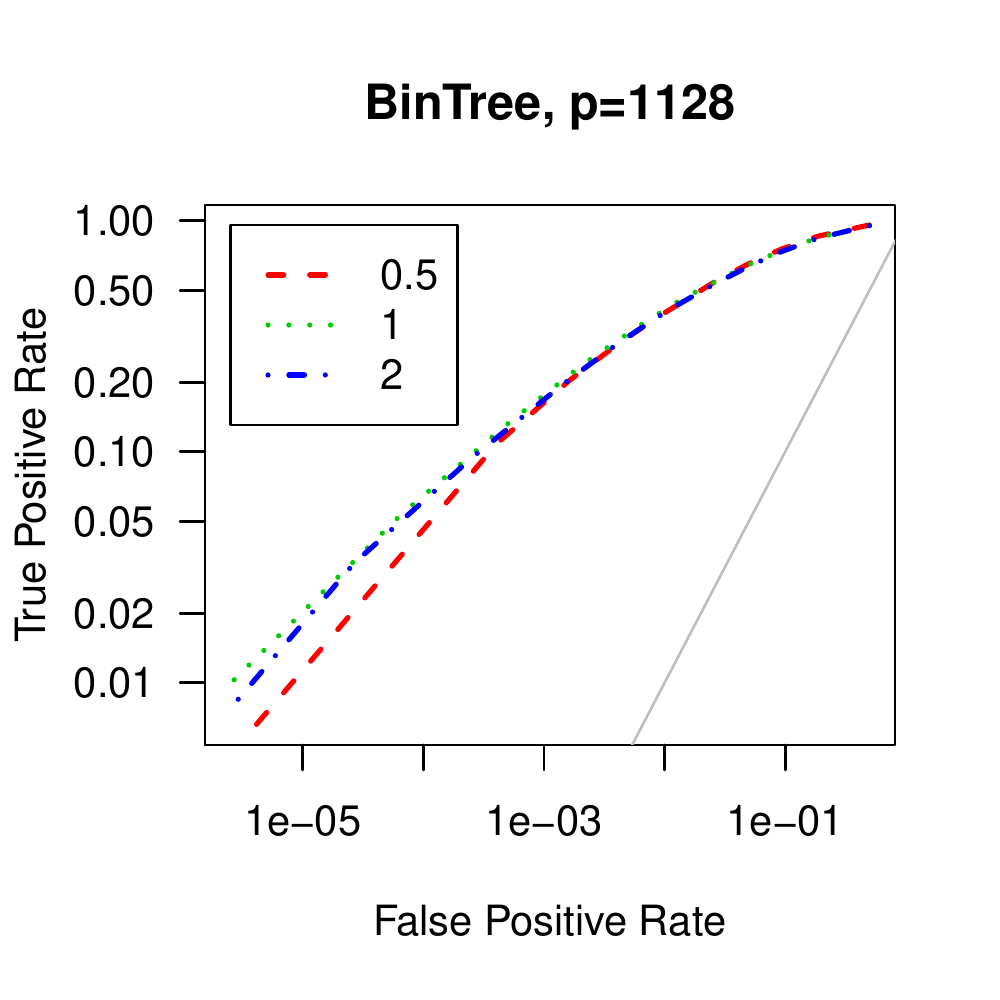}
		\includegraphics[width=0.32\textwidth]{\PICDIR/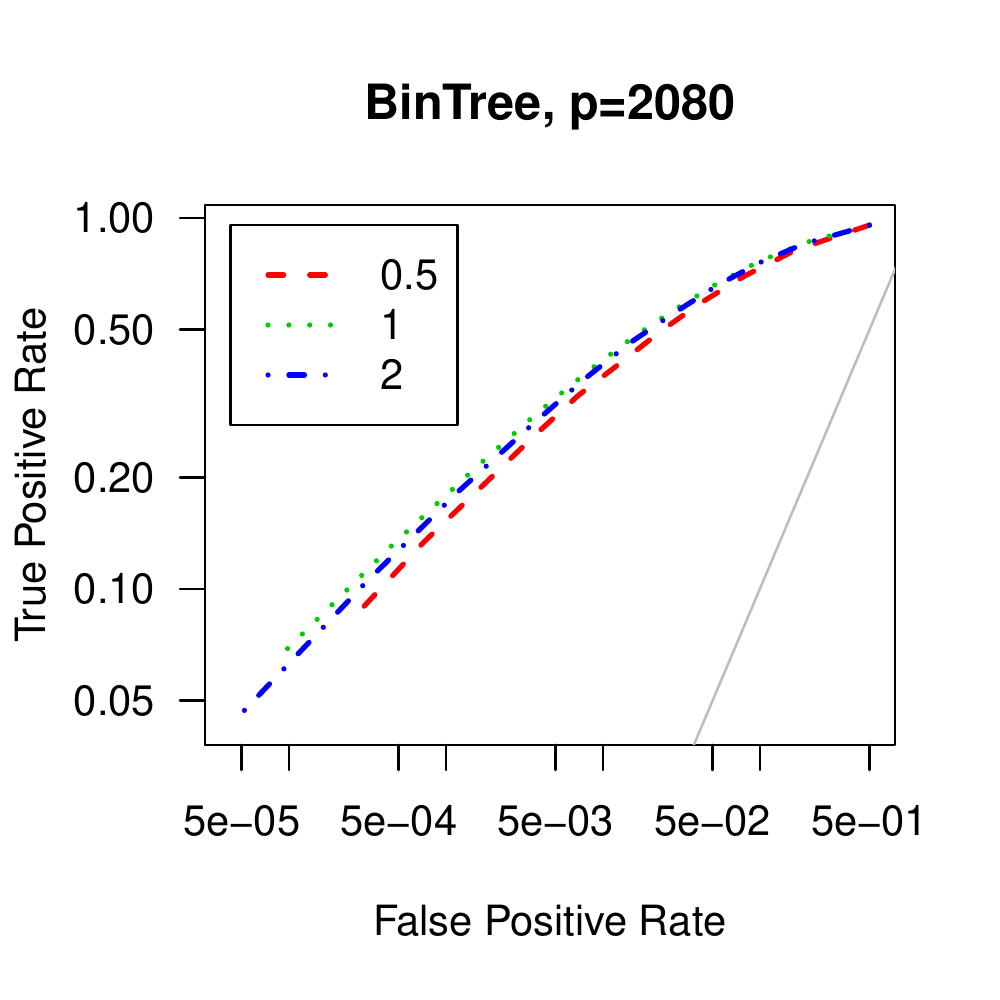}
		\includegraphics[width=0.32\textwidth]{\PICDIR/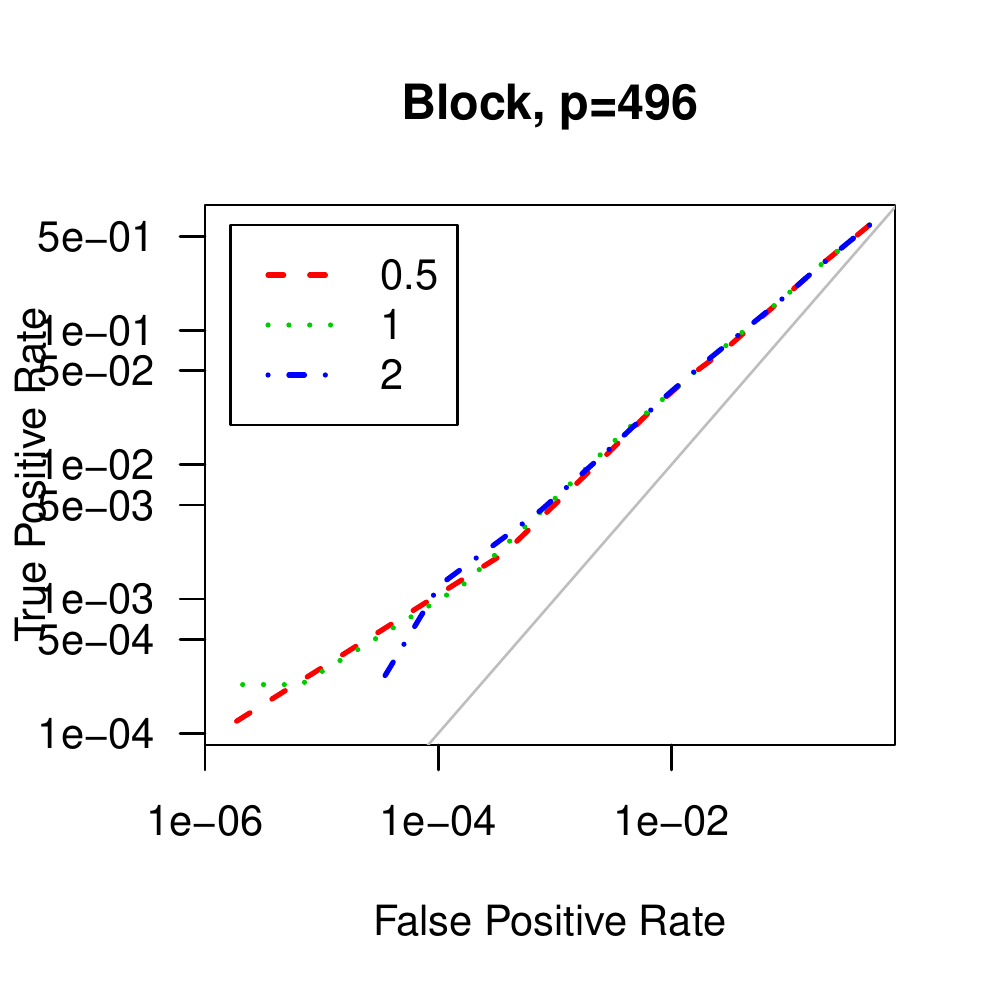}
		\includegraphics[width=0.32\textwidth]{\PICDIR/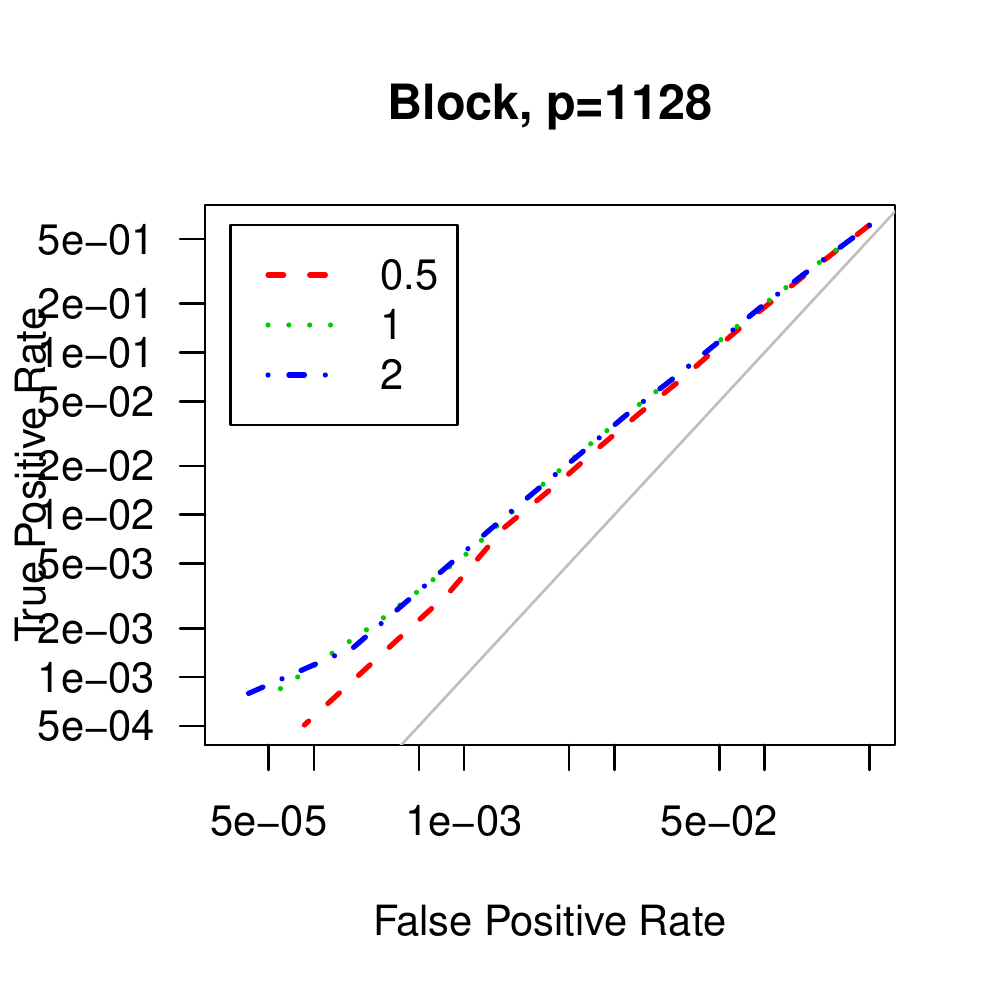}
		\includegraphics[width=0.32\textwidth]{\PICDIR/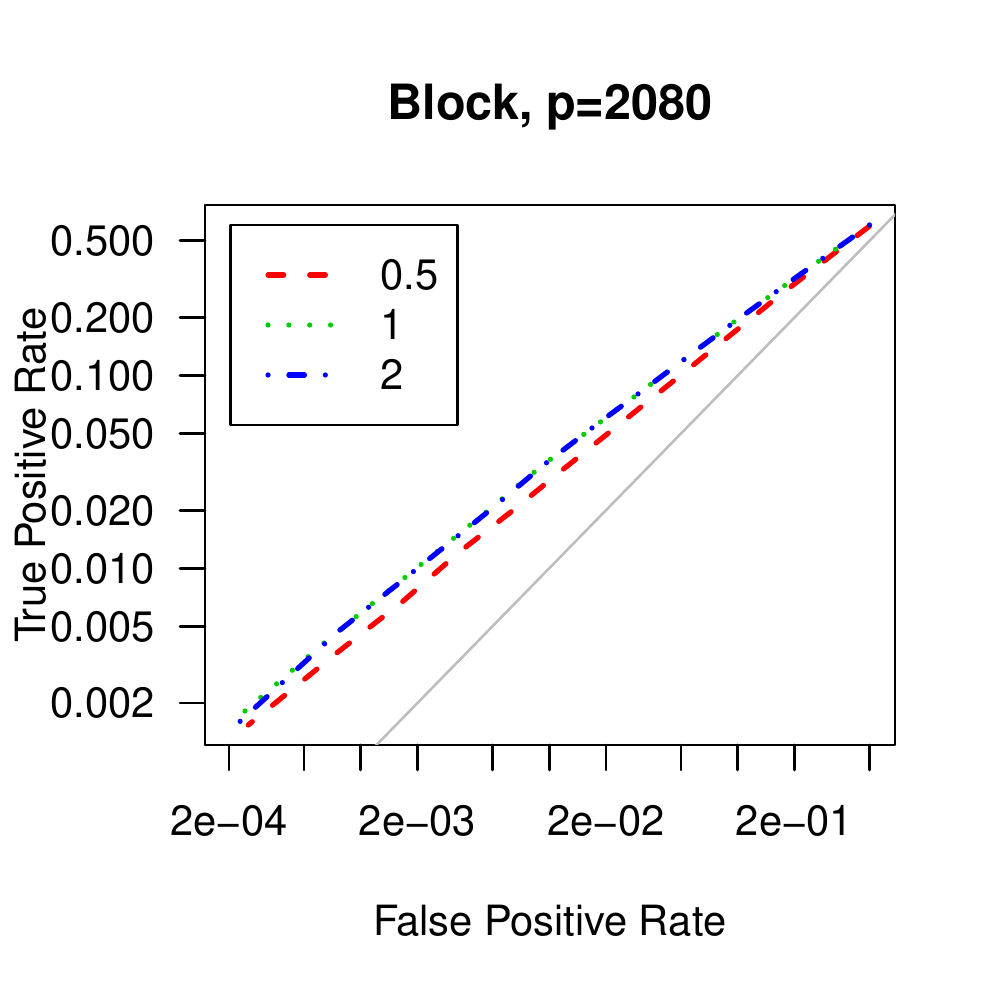}
	\end{center}
	\capt{
		\label{fig:rocs}
		Empirical Receiver operating characteristic
                curves displayed on the log-log scale.
		The rows from top to bottom correspond to the tridiagonal, 
		binary tree, and block diagonal matrices.  The columns from
		left to right correspond to dimensions 496, 1128, and 2080.
	}
\end{figure}

\subsection{Sub-Exponential Data}

The same simulations as in Section~3.1
were 
rerun replacing the multivariate Gaussian distribution with
the multivariate Laplace distribution, and are 
displayed in Figure~\ref{fig:fptpPlotE}.  As expected,
the true positive rate is not as large as in the Gaussian
setting.  When the penalization parameter for the 
graphical lasso is set to $\lmb=2$, we see that our 
methodology does not maintain the desired false positive
rate as $\alpha\rightarrow0$.  However, for $\lmb=1$, the
empirical false positive rate does approximately track 
with the desired false positive rate in the three simulation
settings.

\begin{figure}
  \begin{center}
  \includegraphics[width=0.32\textwidth]{\PICDIR/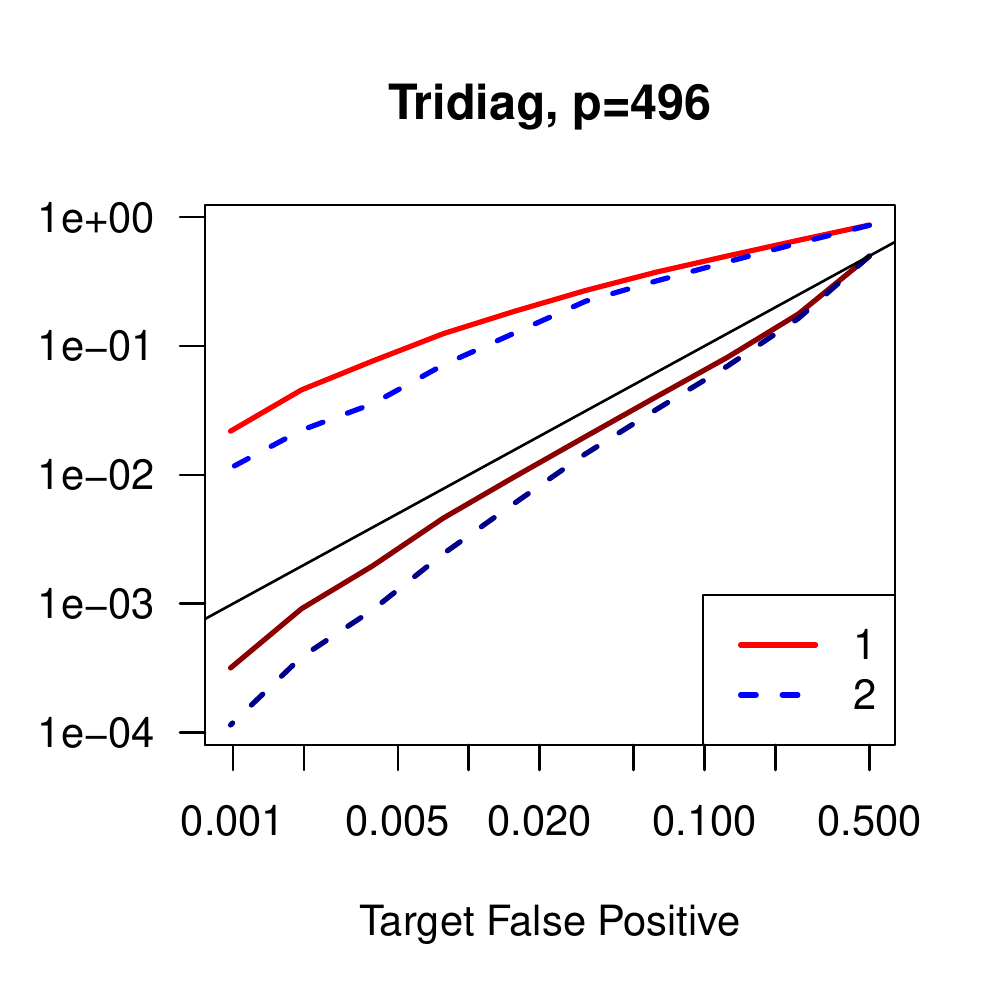}
  \includegraphics[width=0.32\textwidth]{\PICDIR/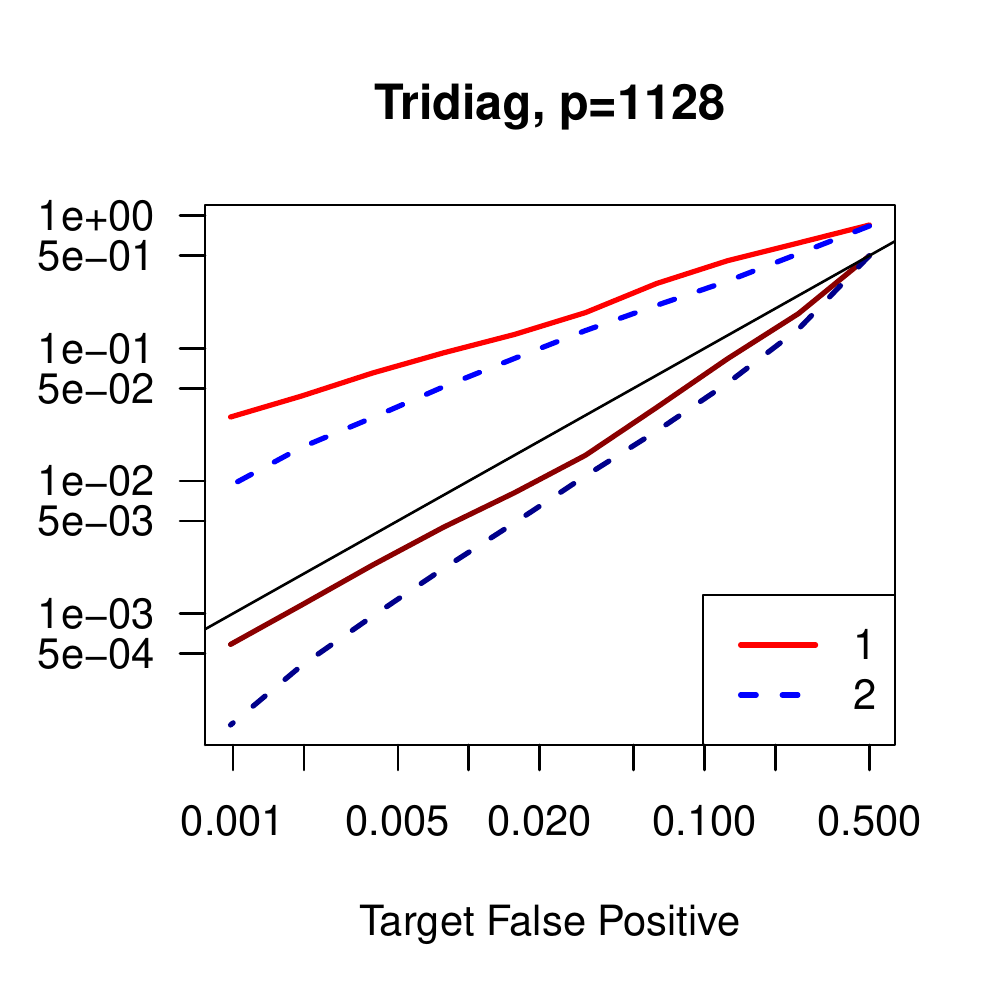}
  \includegraphics[width=0.32\textwidth]{\PICDIR/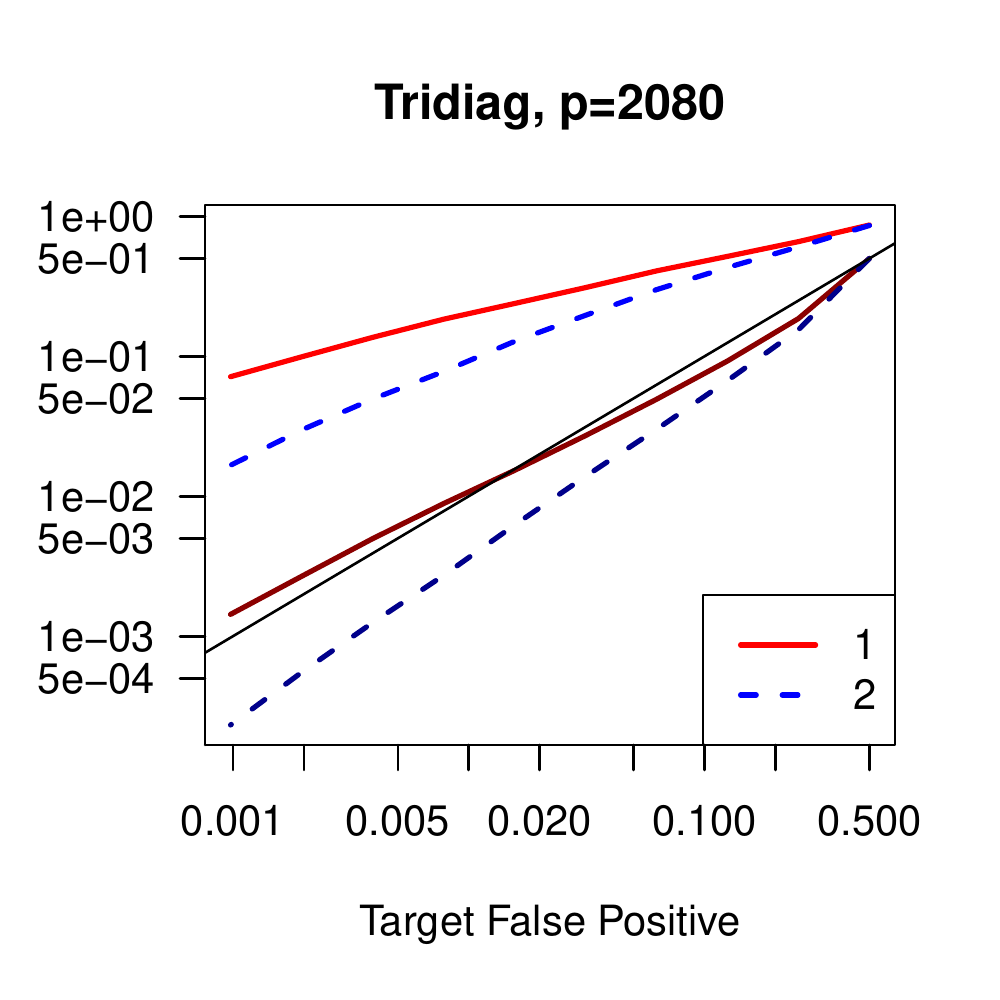}
  \includegraphics[width=0.32\textwidth]{\PICDIR/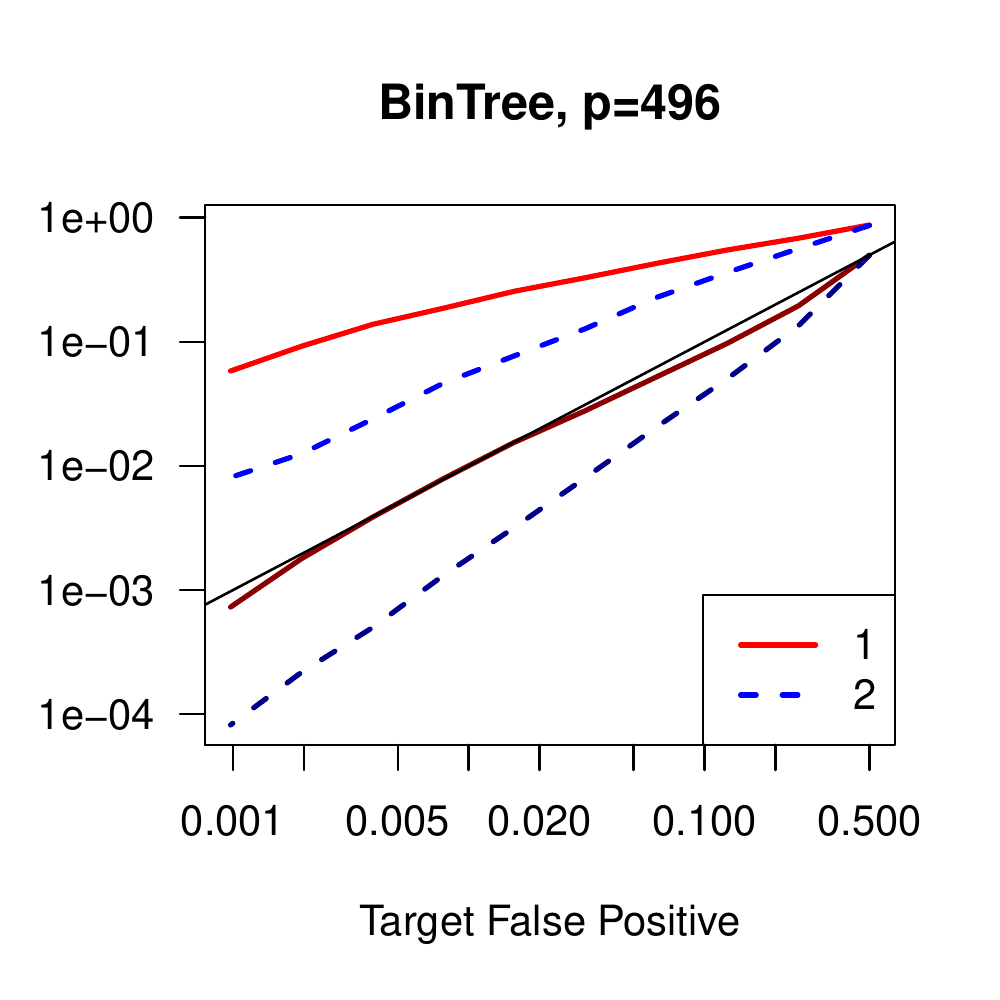}
  \includegraphics[width=0.32\textwidth]{\PICDIR/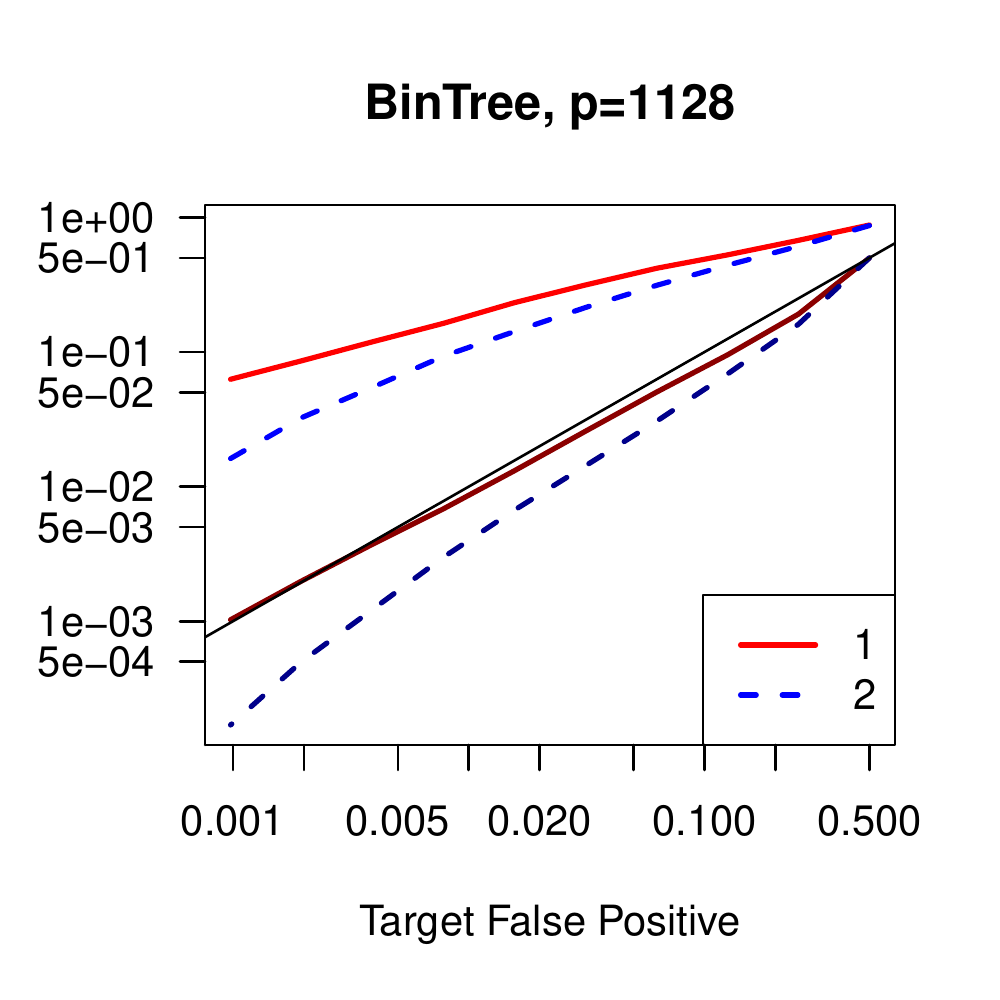}
  \includegraphics[width=0.32\textwidth]{\PICDIR/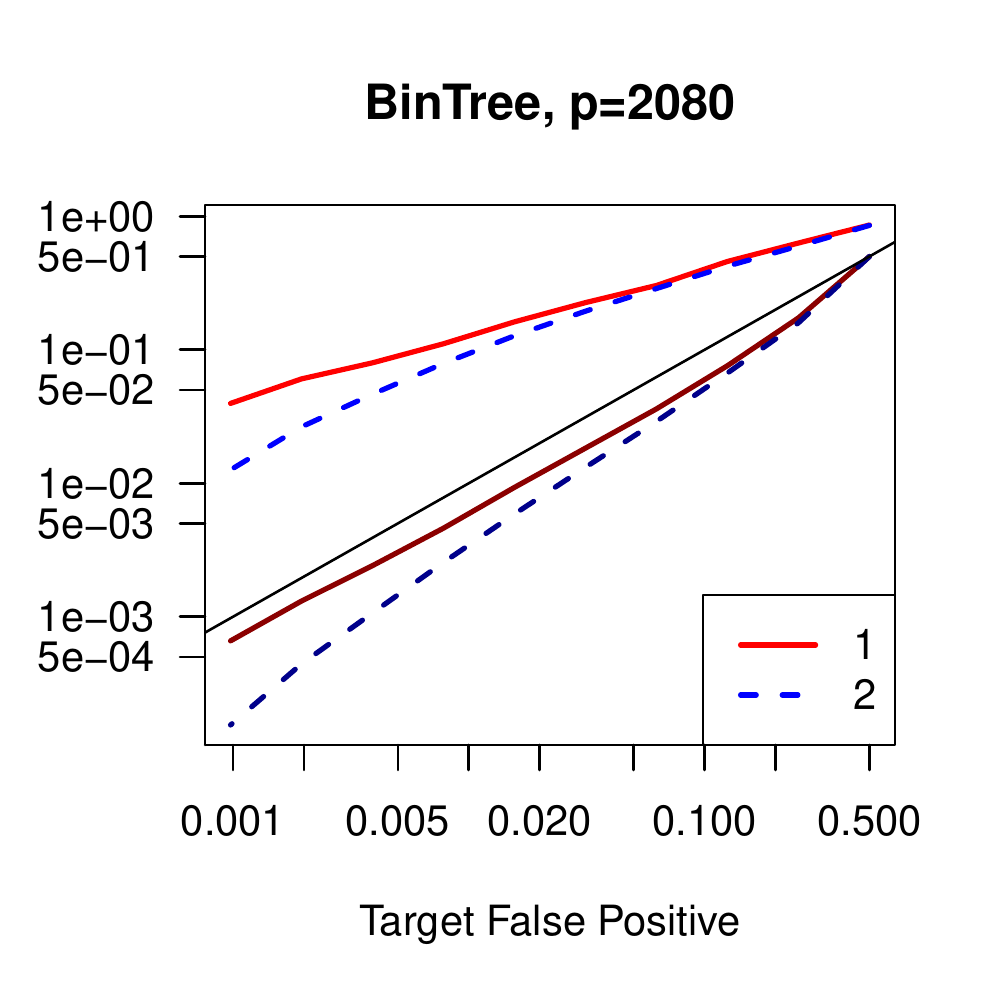}
  \includegraphics[width=0.32\textwidth]{\PICDIR/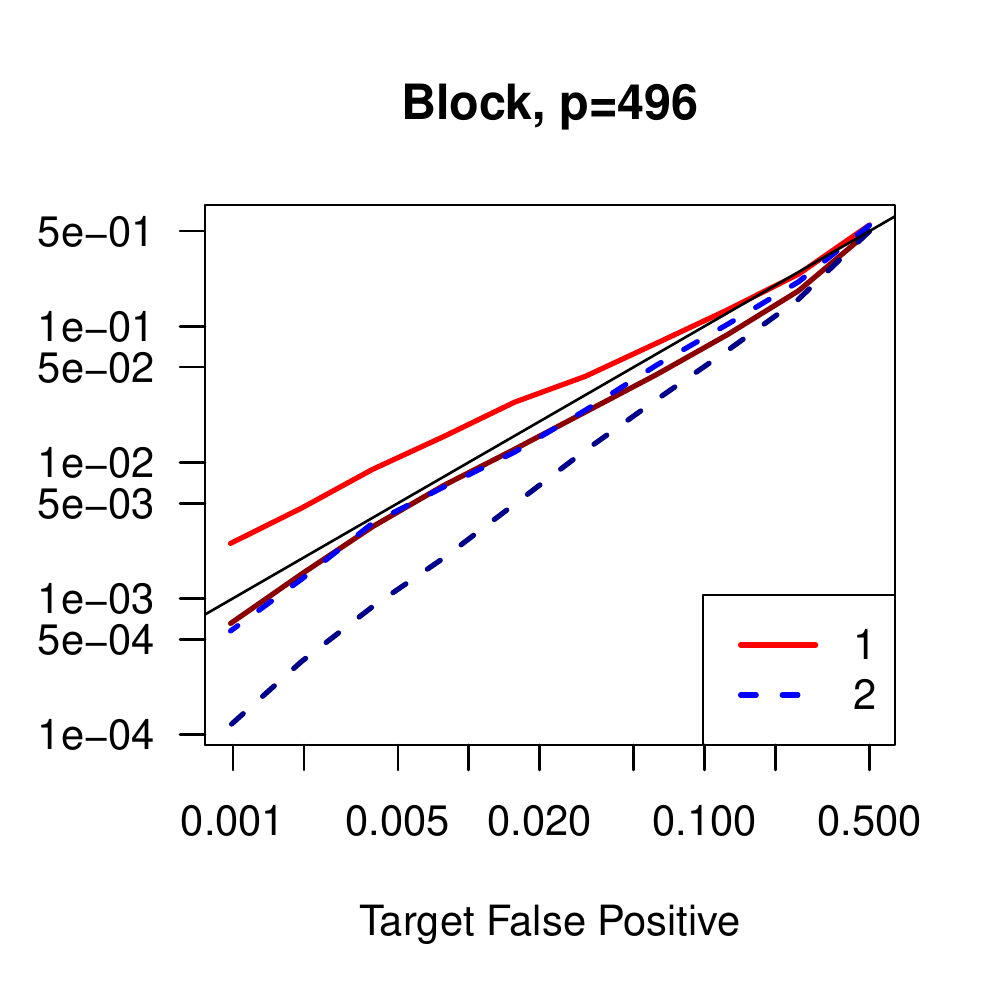}
  \includegraphics[width=0.32\textwidth]{\PICDIR/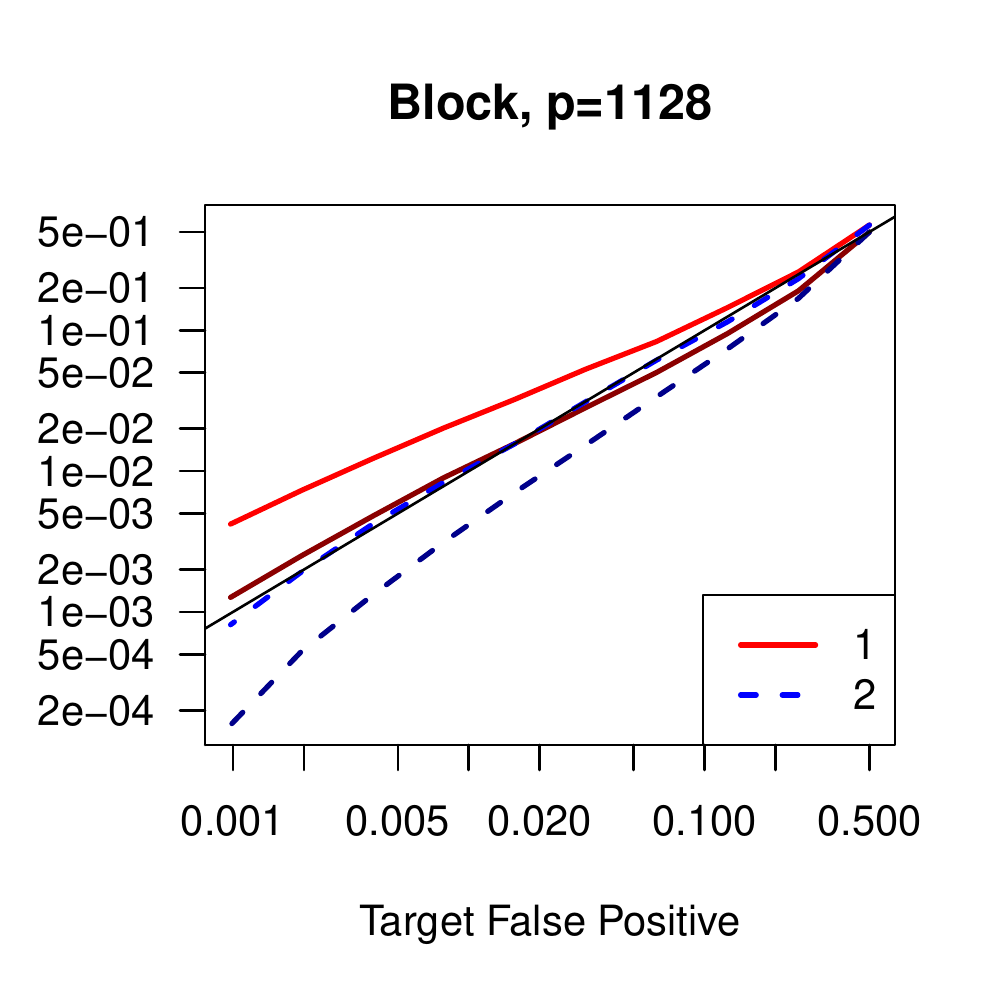}
  \includegraphics[width=0.32\textwidth]{\PICDIR/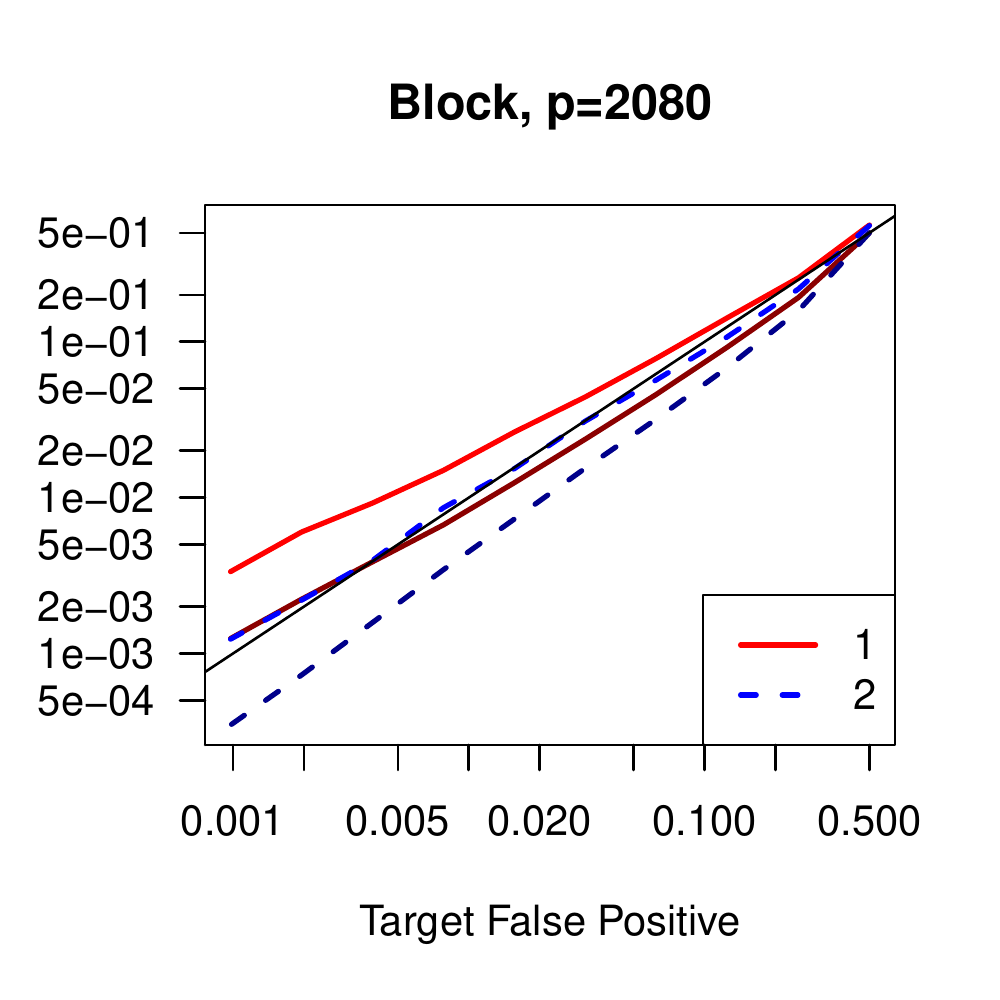}
  \end{center}
  \capt{
    \label{fig:fptpPlotE}
    The achieved true and false positives plotted against the
    target false positive rate displayed on the log-log scale
    for multivariate Laplace data.
    The rows from top to bottom correspond to the tridiagonal, 
    binary tree, and block diagonal matrices.  The columns from
    left to right correspond to dimensions 496, 1128, and 2080.
  }
\end{figure}

\subsection{Asymptotic Threshold}

In the work of \cite{JANKOVA2015}, a thresholding method is proposed 
for the debiased graphical lasso estimator making use of the normal distribution
function.  Namely, 
$$
\hat{\Omega}_{i,j} < 
\Phi^{-1}\left(1-\frac{\alpha}{p(p-1)}\right)
\frac{\hat{\Sigma}_{i,j}}{\sqrt{n}}
$$
where $\hat{\Sigma}$ is the empirical covariance matrix and $\Phi(\cdot)$ is
the cumulative distribution function for the standard normal distribution.  
The same simulations as in Section~3.1
were run on this method
for graphical lasso penalization parameters of $\lmb = 2^{-4},2^{-3},2^{-2},2^{-1},1,2$.
The larger values of $\lmb$ resulted in all of the off-diagonal entries being
set to zero.  Hence, Figure~\ref{fig:jankova} displays the false and true
positive rates for only $\lmb=1/16,1/8$.  This method should work asymptotically
as $n,p\rightarrow\infty$.  However, for our specific choices of $n$ and $p$,
this method failed to achieve anything close to the target false positive rate. 

\begin{figure}
	\begin{center}
		\includegraphics[width=0.32\textwidth]{\PICDIR/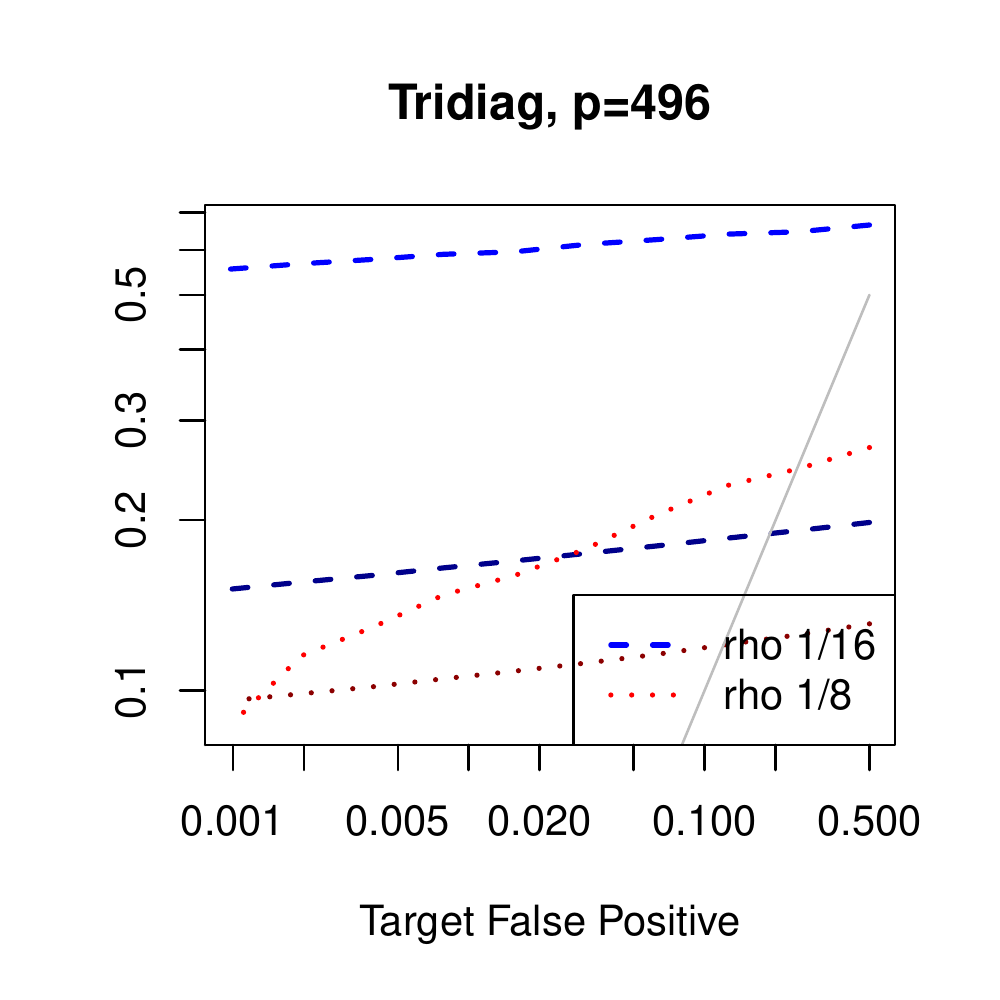}
		\includegraphics[width=0.32\textwidth]{\PICDIR/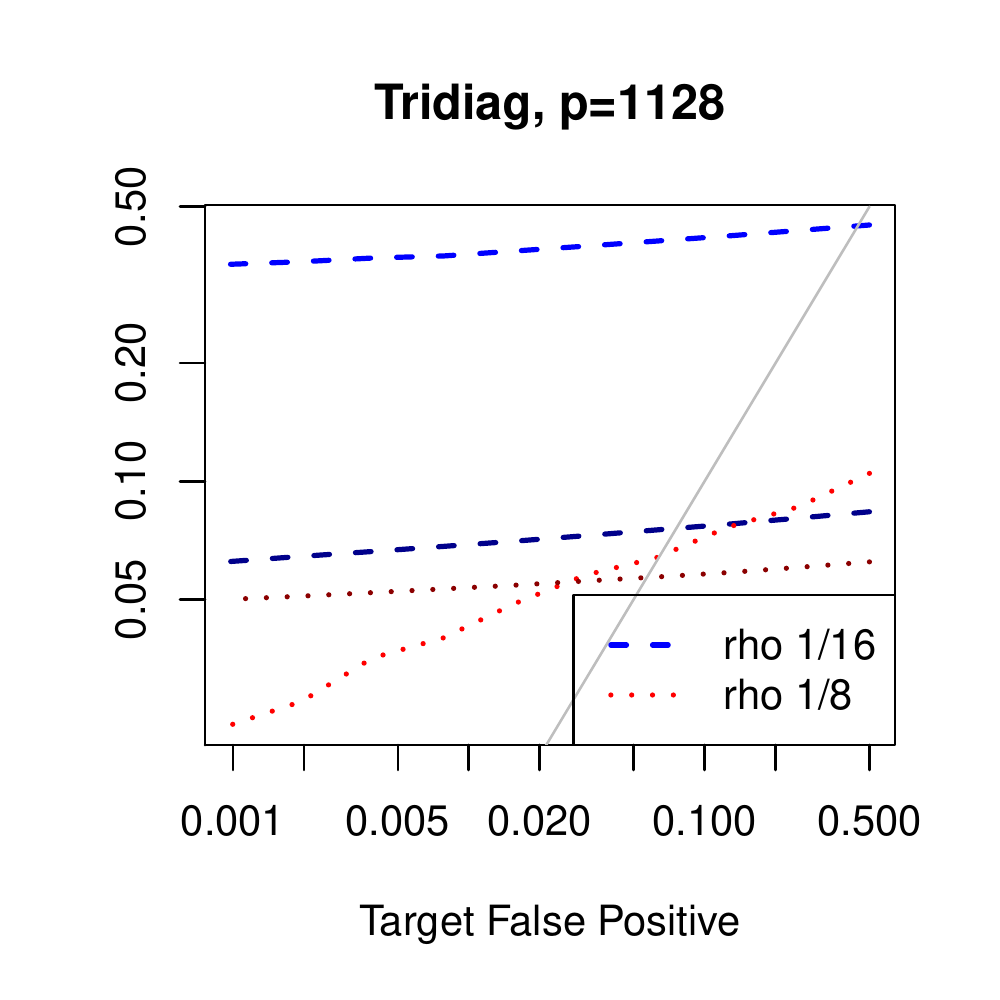}
		\includegraphics[width=0.32\textwidth]{\PICDIR/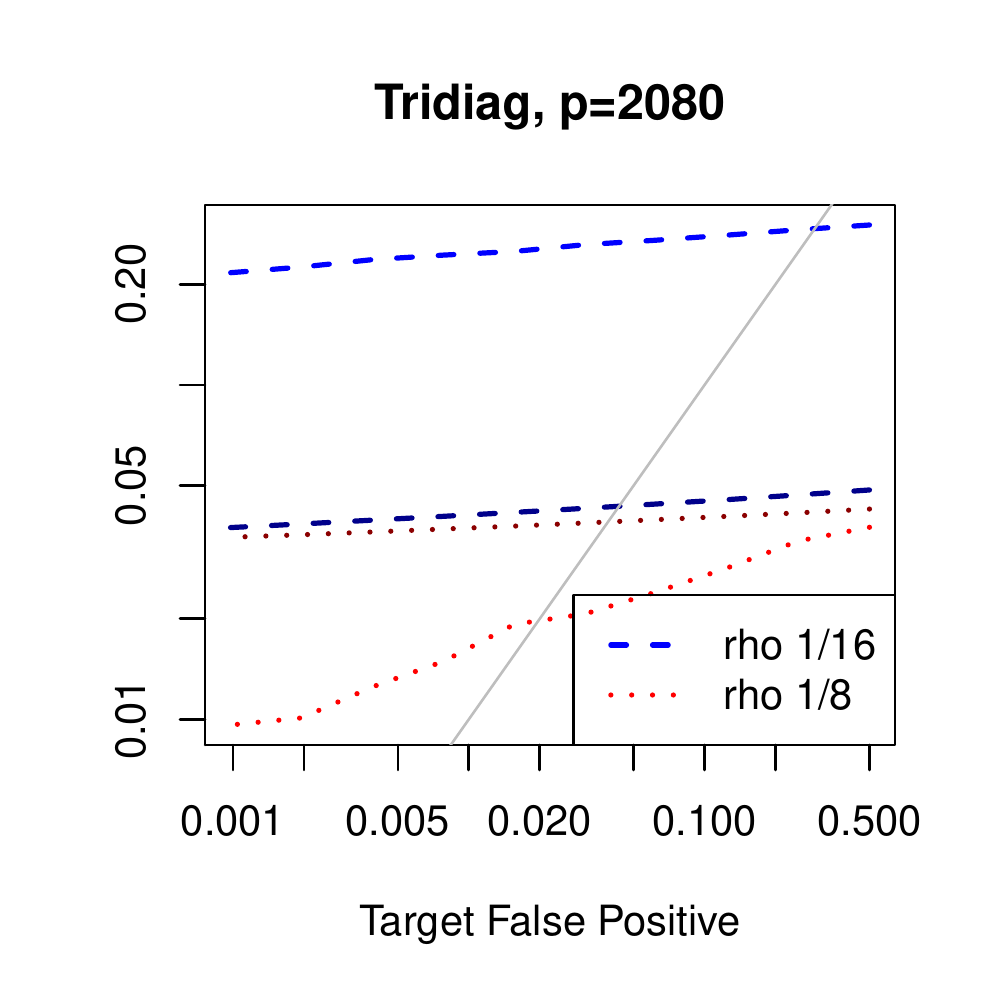}
		\includegraphics[width=0.32\textwidth]{\PICDIR/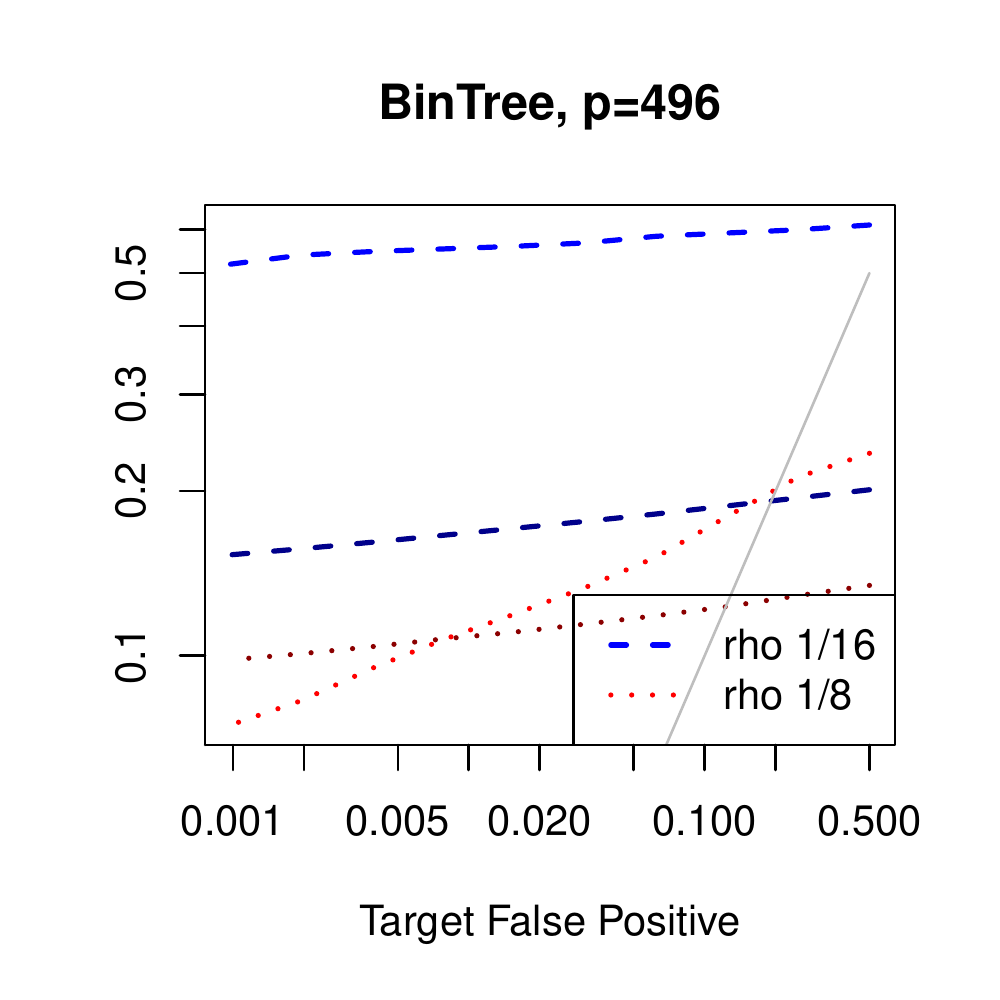}
		\includegraphics[width=0.32\textwidth]{\PICDIR/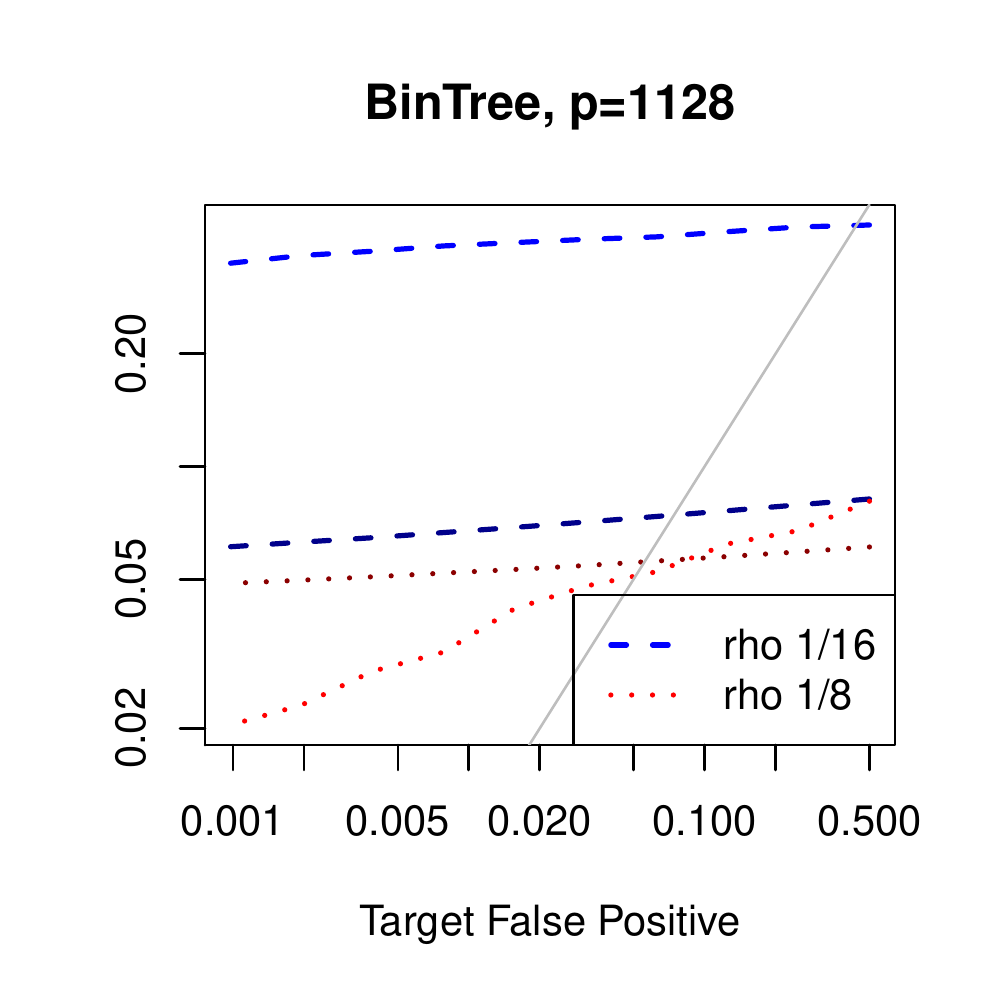}
		\includegraphics[width=0.32\textwidth]{\PICDIR/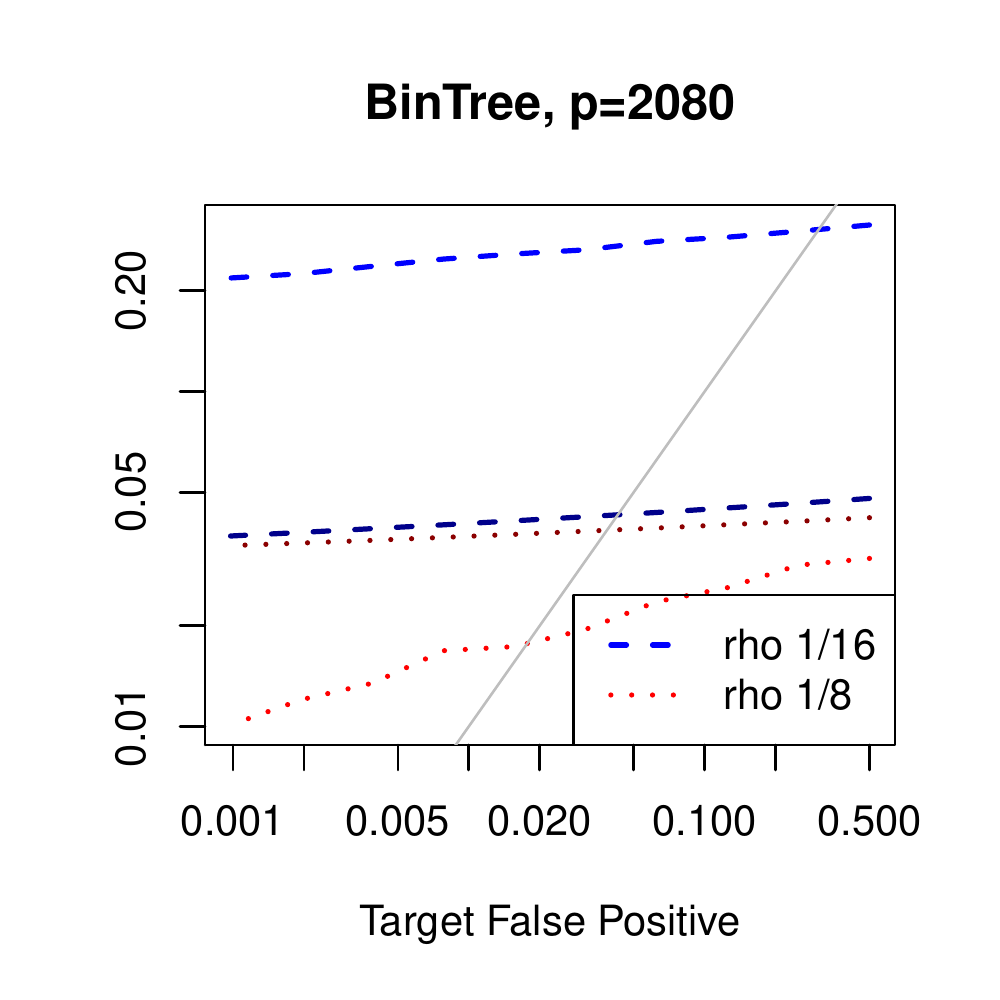}
		\includegraphics[width=0.32\textwidth]{\PICDIR/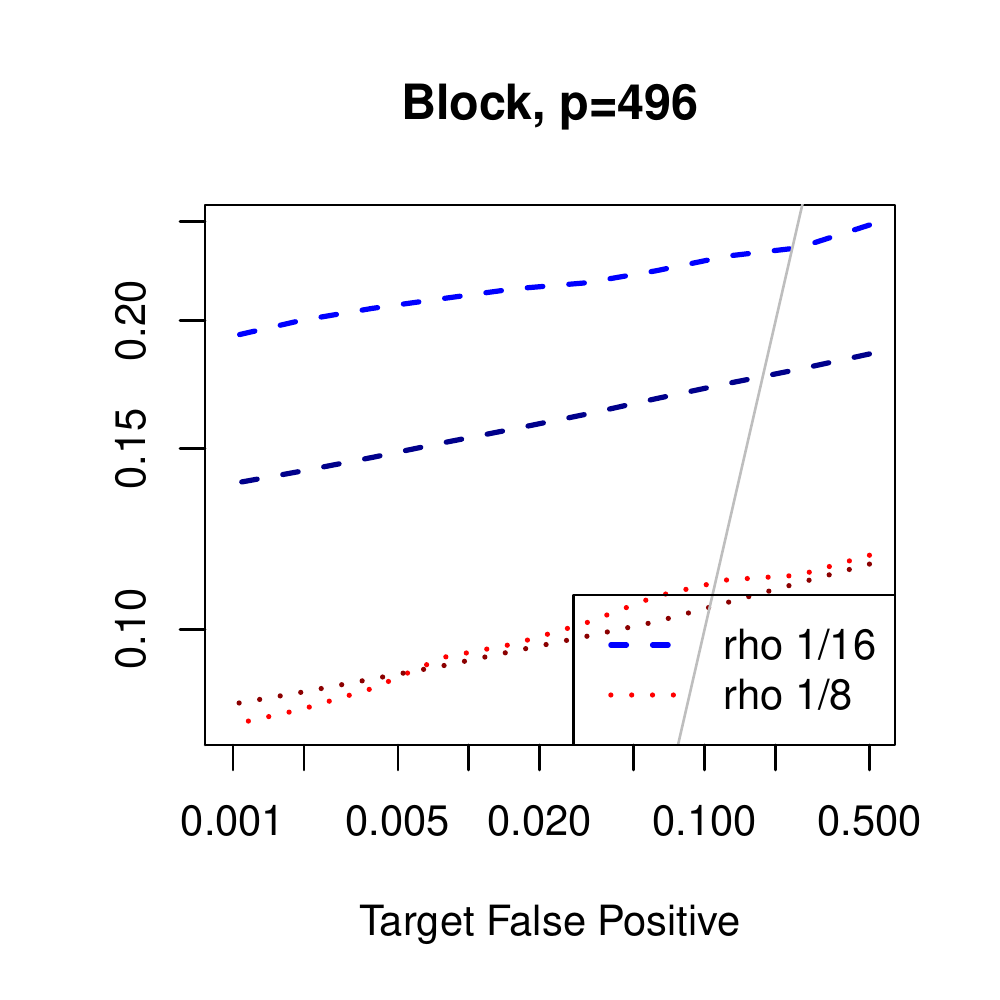}
		\includegraphics[width=0.32\textwidth]{\PICDIR/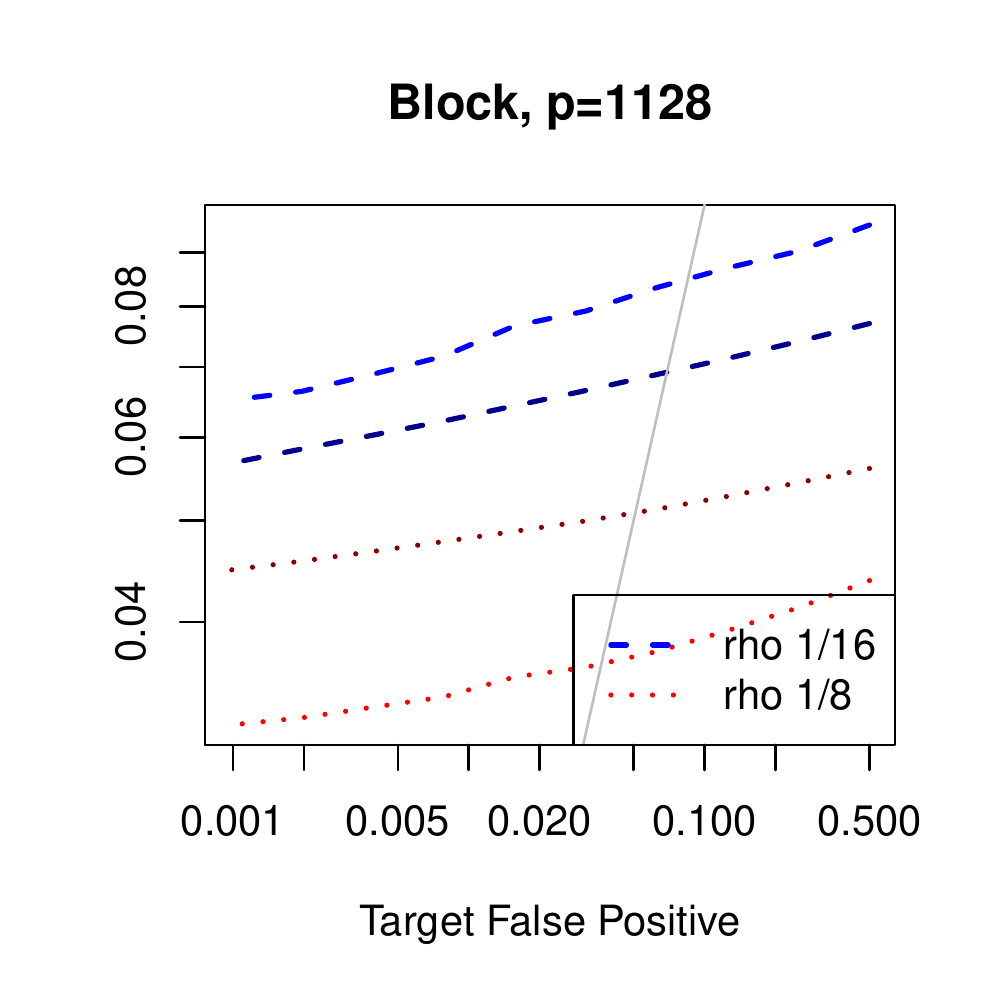}
		\includegraphics[width=0.32\textwidth]{\PICDIR/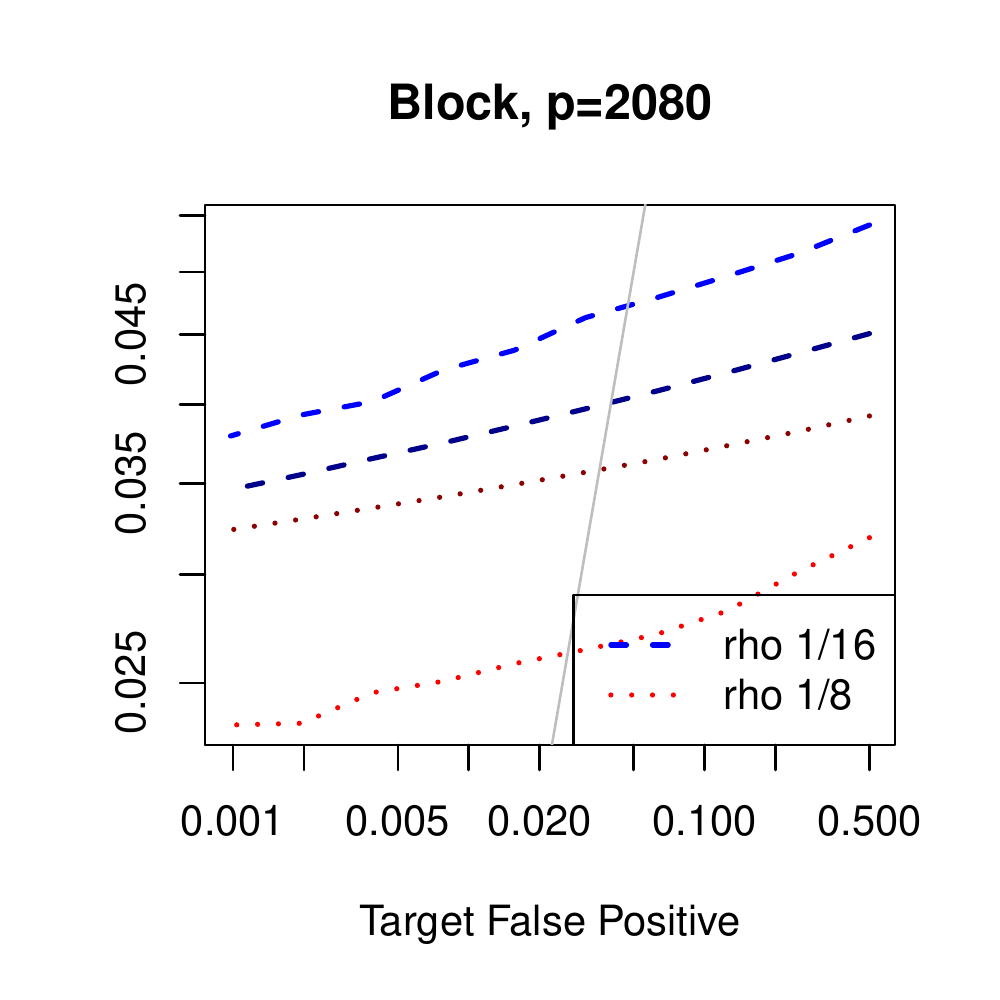}
	\end{center}
	\capt{
		\label{fig:jankova}
		Achieved true and false positive rates plotted against the target 
		rate and displayed on the log-log scale.
		The rows from top to bottom correspond to the tridiagonal, 
		binary tree, and block diagonal matrices.  The columns from
		left to right correspond to dimensions 496, 1128, and 2080.
	}
\end{figure}

\end{document}